\newif\ifarxiv

\arxivtrue % Version with appendix for arXiv
\documentclass[letterpaper]{article} % DO NOT CHANGE THIS

\ifarxiv
  \usepackage{aaai23-nocopyright}  
\else
  \usepackage{aaai23}  % DO NOT CHANGE THIS
\fi

\usepackage{times}  % DO NOT CHANGE THIS
\usepackage{helvet}  % DO NOT CHANGE THIS
\usepackage{courier}  % DO NOT CHANGE THIS
\usepackage[hyphens]{url}  % DO NOT CHANGE THIS
\usepackage{graphicx} % DO NOT CHANGE THIS
\usepackage{natbib}  % DO NOT CHANGE THIS AND DO NOT ADD ANY OPTIONS TO IT
\usepackage{caption} % DO NOT CHANGE THIS AND DO NOT ADD ANY OPTIONS TO IT
\usepackage{algorithm}
\usepackage{algorithmic}

\usepackage{newfloat}
\usepackage{listings}
\DeclareCaptionStyle{ruled}{labelfont=normalfont,labelsep=colon,strut=off} % DO NOT CHANGE THIS
\floatstyle{ruled}
\newfloat{listing}{tb}{lst}{}
\floatname{listing}{Listing}
\graphicspath{{figures/}}

\usepackage{amsmath,amssymb,amsfonts,amsthm}
\usepackage{thm-restate}
\usepackage{booktabs}
\usepackage{algorithm}
\usepackage{algorithmic}

\usepackage{tikz}
\usepackage{color}
\usepackage{soul}

\newtheorem{example}{Example}
\newtheorem{theorem}{Theorem}

\newtheorem{lemma}{Lemma}
\newtheorem{proposition}{Proposition}
\newtheorem{corollary}{Corollary}
\newtheorem{definition}{Definition}

\usepackage[inline]{enumitem}
\newlist{romanenumerate*}{enumerate*}{1}
\setlist[romanenumerate*]{label=(\textit{\roman*})}
\newlist{romanenumerate}{enumerate}{1}
\setlist[romanenumerate]{label=(\textit{\roman*})}

\usepackage{mymacros}

\title{Automata Cascades: Expressivity and Sample Complexity}

\author {
    Alessandro Ronca,\textsuperscript{\rm 1}
    Nadezda Alexandrovna Knorozova,\textsuperscript{\rm 2,3}
    Giuseppe De Giacomo\textsuperscript{\rm 1,4}
}
\affiliations {
    \textsuperscript{\rm 1}DIAG, Sapienza University of Rome\\
    \textsuperscript{\rm 2}RelationalAI\\
    \textsuperscript{\rm 3}IFI, University of Zurich\\
    \textsuperscript{\rm 4}Computer Science Department, University of Oxford\\
    ronca@diag.uniroma1.it, 
    nadezda.knorozova@relational.ai, 
    giuseppe.degiacomo@cs.ox.ac.uk
}

\newcounter{runningexamplecounter}
\begin{document}

\maketitle

%%% Sections

\begin{abstract}
  Every automaton can be decomposed into a cascade of basic \emph{prime}
  automata.
  This is the Prime Decomposition Theorem by Krohn and Rhodes.
  Guided by this theory,
  we propose \emph{automata cascades} as a structured, modular,
  way to describe automata as complex systems made of many components, each
  implementing a specific functionality.
  Any automaton can serve as a component; using specific components allows for
  a fine-grained control of the expressivity of the resulting class of
  automata;
  using prime automata as components implies specific expressivity guarantees.
  % We interpret the theorem dually seeing cascades as a way to express
  % automata as the \emph{composition} of basic automata, each implementing a
  % specific functionality.
  % Then the theorem by Krohn and Rhodes prescribes which components to use in
  % order to obtain classes of automata with a specific expressivity.
  % % We interpret the theorem dually, as a theorem that prescribes which prime
  % % automata to use in order to obtain classes of automata with a specific
  % % expressivity.
  Moreover, specifying automata as cascades allows for
  describing the sample complexity of automata in terms of their components.
  We show that the sample complexity is linear in the number of
  components and the maximum complexity of a single component, modulo
  logarithmic factors.
  This opens to the possibility of learning automata representing large
  dynamical systems consisting of many parts interacting with each other.
  It is in sharp contrast with the established understanding of the sample
  complexity of automata, described in terms of the overall number of states and
  input letters, which implies that it is only possible to learn
  automata where the number of states is linear in the amount of data available.
  Instead our results show that one can learn automata with a number of states
  that is exponential in the amount of data available.
\end{abstract}

\section{Introduction}

Automata are fundamental in computer science.  
They are one of the simplest models of computation, with the expressive
power of regular languages, placed at the bottom of the Chomsky hierarchy.
They are also a mathematical model of finite-state dynamical systems.
In learning applications, automata allow for capturing targets that exhibit a
time-dependent behaviour; namely, functions over sequences such
as time-series, traces of a system, histories of interactions between an agent
and its environment.
Automata are typically viewed as \emph{state diagrams}, where states and
transitions are the building blocks of an automaton, cf.\
\cite{hopcroft1979ullman}. 
Accordingly, classes of automata are described in terms of the number of states
and input letters.
Learning from such classes requires an amount of data that is linear in the
number of states and letters~\cite{ishigami1997vc}.
Practically it means that, in order to learn large dynamical systems made of
many components, the amount of data required is exponential in the number of
components, as the number of states of a system will typically be exponential in
the number of its stateful components.

We propose \emph{automata cascades} as a structured, modular, way to describe
automata as complex systems made of many \emph{components} connected in an
acyclic way.
Our cascades are strongly based on the \emph{theory of Krohn and Rhodes},
which says that every automaton can be decomposed into a cascade of basic
components called \emph{prime automata} \cite{krohn1965rhodes}. 
Conversely, the theory can be seen as prescribing which components to use in
order to build certain classes of automata, and hence obtain a certain
expressivity.
For example, we can cascade so-called \emph{flip-flop automata} in
order to build all noncounting automata, and hence obtain the expressivity of
well-known logics such as monadic first-order logic on finite
linearly-ordered domains \cite{mcnaughton1971counter} and
the linear temporal logic on finite traces 
LTL$_f$ \cite{degiacomo2013ltlf}.

We focus on cascades as a means to learn automata. 
Our cascades are designed for a fine control of their sample complexity. 
% For
% instance, they are equipped with an input-processing mechanism that allows for
% operating on large (even infinite) input alphabets without directly affecting
% the complexity of the core stateful mechanism of a component.
% % Having a richer structure, cascades allow for describing classes of automata in
% % a richer way than simply specifying the number of states and input letters.
% % Namely, one can design classes by choosing the number and kind of components. In
% % turn, components and the
% % kinds of components they allow for describing classes of automata in terms of their
% % components. One can design classes by choosing the prime automata combinations 
% % 
% Our main contribution is an analysis of the sample complexity of automata
% cascades.
We show that---ignoring logarithmic factors---\emph{the sample complexity of
  automata cascades is at most linear in the product of the number of components
  and the maximum complexity of a single component}.
Notably, the complexity of a single component does not grow with 
the number of components in the cascade. 
% Fixing it yields an overall
% sample complexity that grows at most linearly with the number of cascade
% components.
% modulo logarithmic factors.
% In particular, we can  have large cascae cascades can grows 
% In particular, for cascades built out of $d$ components of a certain complexity,
% the sample complexity grows at most with $d \cdot \log d$ linearly with the number of components.
% For cascades made of components of fixed complexity
% In particular, it is linear in the number components, keeping the rest fixed.
We carry out the analysis both in the setting where classes of
automata cascades are finite, and in the more general setting where they can be
infinite.
For both cases, we obtain bounds of the same shape, with one notable difference
that for infinite classes we incur a logarithmic dependency on the
maximum length of a string.
Overall, our results show that the sample complexity of automata can be
decoupled from the the number of states and input letters. Rather, it can be
described in terms of the components of a cascade capturing the automaton.
Notably, the number of states of such an automaton can be
exponential in the number of components of the cascade.

% In the analysis, we consider both finite and infinite classes of automata
% cascades.
% Finite classes occur in settings where the input alphabet is
% finite---e.g., input letters consisting of Boolean features.
% Infinite classes arise naturally when the input alphabet is infinite---e.g.,
% integer input letters or real-valued vectors. 
% The two settings are relevant in different application scenarios.
% For both cases, we obtain bounds of the same shape, with one notable difference
% that for infinite classes we incur a logarithmic dependency on the
% maximum length of a string.
% Overall, our results show that the sample complexity of automata can be
% decoupled from the the number of states and input letters. Rather, it can be
% described in terms of the components of a cascade capturing the automaton.
% Notably, the number of states of such an automaton can be
% exponential in the number of components of the cascade.

We see the opportunity for cascades to unlock a greater potential of 
automata in learning applications. 
On one hand, automata come with many favourable, well-understood, 
theoretical properties, and they admit elegant algorithmic solutions. 
On the other hand, the existing automata learning algorithms have a complexity
that depends directly on the number of states.
This hurts applications where the number of states grows very fast, such as 
non-Markov reinforcement learning
\cite{toroicarte2018reward,degiacomo2019bolts,gaon2019rl,brafman2019rdp,abadi2020learning,xu2020jirp,neider2021advice,jothimurugan2021compositional,ronca2021pac,ronca2022markov}.
Given our favourable sample complexity results, automata cascades have a great
potential to extend the applicability of automata learning in large complex
settings.

% potential to enable automata learning in large complex settings.

Before concluding the section, we introduce our running example, that is
representative of a class of tasks commonly considered in reinforcement
learning.
It is based on an example from \cite{andreas2017modular}.

% \begin{figure}[t]
%   \centering
%   \vspace{0.1cm}
%   \includegraphics[draft=false,width=0.35\textwidth]{task-diagram-for-example}
%   \caption{Task diagram for the running example.}
%   \label{figure:task-diagram}
% \end{figure}

\begin{example}% [Running example, Part~\Roman{runningexamplecounter}] 
  % \stepcounter{runningexamplecounter}
  \label{ex:running-example-1}
  Consider a Minecraft-like domain,
  where an agent has to build a bridge by first collecting 
  some raw materials, and then using a factory.
  The agent has two options.
  The first option is to complete  
  the tasks $\{\mathrm{getWood},$ $\mathrm{getIron},$ $\mathrm{getFire} \}$ in any
  order, and then $\{ \mathrm{useFactory} \}$.
  The second option is to complete $\{ \mathrm{getSteel} \}$, and then
  $\{ \mathrm{useFactory} \}$.
  % The task diagram is shown in Figure~\ref{figure:task-diagram}.
\end{example}

\ifarxiv
  A formal description of the example as well as full proofs of all our
  technical results are deferred to the appendix.
\else
  A formal description of the example as well as full proofs of all our
  technical results are given in the extended version \cite{extendedversion}.
\fi

\section{Preliminaries}

\paragraph{Functions.}
For $f: X \to Y$ and $g: Y \to Z$, their \emph{composition}
$f \circ g: X \to Z$ is defined as $(f \circ g)(x) = g(f(x))$. 
For $f: X \to Y$ and $h: X \to Z$, their \emph{cross product} 
$f \times h: X \to Y \times Z$ is defined as 
$(f \times h)(x) = \langle f(x), h(x) \rangle$. 
A class of functions is \emph{uniform} if all functions in the class
have the same domain and codomain.
Given a tuple $t = \langle x_1, \dots, x_n \rangle$, and 
a subset $J \subseteq [1,n]$,
the \emph{projection} $\pi_J(t)$ is
$\langle x_{j_1}, \dots, x_{j_m} \rangle$ where  
$j_1, \dots, j_m$ is the sorted sequence of elements of $J$.
Furthermore, $\pi_m^a$ denotes the class of all projections $\pi_J$ for $J$ a
subset of $[1,a]$ of cardinality $m$.
We write $I$ for the \emph{identity function}, and $\log$ for $\log_2$.

\paragraph{String Functions and Languages.}
An \emph{alphabet} $\Sigma$ is a set of elements called \emph{letters}.
A \emph{string} over $\Sigma$ is an expression $\sigma_1 \dots \sigma_\ell$
where each letter $\sigma_i$ is from $\Sigma$.
The \emph{empty string} is denoted by $\varepsilon$.
The set of all strings over $\Sigma$ is denoted by $\Sigma^*$.
A \emph{factored alphabet} is of the form $X^a$ for some set $X$ called the
\emph{domain} and some integer $a \geq 1$ called the \emph{arity}.
A \emph{string function} is of the form $f: \Sigma^* \to \Gamma$ for $\Sigma$
and $\Gamma$ alphabets.
\emph{Languages} are a special case of string functions; namely, 
when $f: \Sigma \to \{0,1\}$ is an indicator function,
it can be equivalently described by the set 
$\{ x \in \Sigma^* \mid f(x) = 1 \}$, that is called a \emph{language}.

\section{Learning Theory}

We introduce the problem of learning, following
the classical perspective of statistical learning theory 
\cite{vapnik1998statistical}.

\paragraph{Learning Problem.}
Consider an input domain $X$ and an output domain $Y$. For example, in the
setting where we classify strings over an alphabet $\Sigma$, the input domain
$X$ is  $\Sigma^*$ and the output domain $Y$ is $\{ 0,1 \}$.
The \emph{problem of learning} is that of choosing, from an 
\emph{admissible class} $\mathcal{F}$ of functions from $X$ to $Y$, a function
$f$ that best approximates an unknown \emph{target function} $f_0: X \to Y$, not
necessarily included in $\mathcal{F}$.
The quality of the approximation of $f$ is given by the overall discrepancy of 
$f$ with the target $f_0$. On a single domain element $x$, the discrepancy
between $f(x)$ and $f_0(x)$ is measured as $L(f(x),f_0(x))$, for a given 
\emph{loss function} $L: Y \times Y \to \{ 0,1 \}$.
The overall discrepancy is the expectation $\mathbb{E}[L(f(x),f_0(x))]$ with
respect to an underlying probability distribution $P$,
and it is called the \emph{risk} of the function, written $R(f)$.
Then, the goal is to choose a function $f \in \mathcal{F}$ that
minimises the risk $R(f)$, when the underlying probability distribution $P$ is
unknown, but we are given a \emph{sample} $Z_\ell$ of $\ell$ i.i.d.\ elements
$x_i \in X$ drawn according to $P(x_i)$ together with their labels $f_0(x_i)$;
specifically, $Z_\ell = z_1, \dots, z_\ell$ with 
$z_i = \langle x_i,f_0(x_i) \rangle$.

\paragraph{Sample Complexity.}
We would like to establish the minimum sample size $\ell$ sufficient to identify
a function $f \in \mathcal{F}$ such that 
\begin{equation*}
\label{eq:consistency}
R(f) - \min_{f \in \mathcal{F}} R(f) \leq \epsilon
\end{equation*}
with probability at least $1-\eta$.
We call such $\ell$ the sample complexity of $\mathcal{F}$, and we write it as
$S(\mathcal{F},\epsilon,\eta)$.
When $\epsilon$ and $\eta$ are considered fixed, we write it as $S(\mathcal{F})$.

\paragraph{Sample Complexity Bounds for Finite Classes.}

When the set of admissible functions $\mathcal{F}$ is finite, 
its sample complexity can be bounded in terms of its cardinality, cf.\ 
\cite{shalevshwartz2014book}. In particular, 
\begin{equation*}
  S(\mathcal{F},\epsilon,\eta) \in O \big( (\log |\mathcal{F}| - \log \eta) /
  \epsilon^2 \big).
\end{equation*}

Then, for fixed $\epsilon$ and $\eta$, the sample complexity $S(\mathcal{F})$ is 
$O(\log|\mathcal{F}|)$, and hence finite classes can be compared in terms of
their cardinality.

\section{Automata}
This section introduces basic notions of automata theory, with some inspiration
from \cite{ginzburg,maler1993thesis}.
% In some aspects, our perspective is different from the traditional
% one found in, e.g., \cite{hopcroft1979ullman}.

An automaton is a mathematical description of a stateful machine that returns an
output letter on every input string. 
At its core lies the mechanism that updates the internal state upon reading an
input letter. This mechanism is captured by the notion of semiautomaton.
An $n$-state \emph{semiautomaton} is a tuple 
$D = \langle \Sigma, Q, \delta, q_\mathrm{init} \rangle$ where:
$\Sigma$ is an alphabet called the \emph{input alphabet};
$Q$ is a set of $n$ elements called \emph{states}; 
$\delta: Q \times \Sigma \to Q$ is a function called \emph{transition function}; 
$q_\mathrm{init} \in Q$ is called \emph{initial state}.
The transition function is recursively extended to non-empty strings as
$\delta(q,\sigma_1\sigma_2 \dots \sigma_m) = 
\delta(\delta(q,\sigma_1), \sigma_2 \dots\sigma_m)$, and to the empty string as 
$\delta(q,\varepsilon) = q$.
The result of executing semiautomaton $D$ on an input string is
$D(\sigma_1 \dots \sigma_m) = 
\delta(q_\mathrm{init},\sigma_1 \dots \sigma_m)$.
We also call such function $D: \Sigma^* \to Q$ the function implemented by
the semiautomaton.

Automata are obtained from semiautomata by adding an output function.
An $n$-state \emph{automaton} is a tuple 
$A = \langle \Sigma, Q, \delta, q_\mathrm{init}, \Gamma, \theta \rangle$ where:
$D_A = \langle \Sigma, Q, \delta, q_\mathrm{init} \rangle$ is a semiautomaton,
called the \emph{core semiautomaton} of $A$;
$\Gamma$ is an alphabet called \emph{output alphabet},
$\theta: Q \times \Sigma \to \Gamma$ is called \emph{output function}.
An \emph{acceptor} is a special kind of automaton where the output function is
an indicator function $\theta: Q \times \Sigma \to \{ 0, 1\}$.
The result of executing autmaton $A$ on an input string is
$A(\sigma_1 \dots \sigma_m) = 
\theta(D_A(\sigma_1 \dots \sigma_{m-1}),\sigma_m)$.
We also call such function $A: \Sigma^* \to \Gamma$ the function implemented by
the automaton.
The language \emph{recognised} by an acceptor is the set of strings on which it
returns $1$. 
When two automata $A_1$ and $A_2$ implement the same function, we say that $A_1$
\emph{is captured by} $A_2$, or equivalently that $A_2$ is captured by $A_1$.

The \emph{expressivity} of a class of automata is the set of functions they
implement.
Thus, for acceptors, it is the set of languages they recognise.
The expressivity of all acceptors is the \emph{regular languages}
\cite{kleene1956representation}, i.e., the languages that can be specified by
regular expressions. 
We will often point out the expressivity of acceptors because the classes of
languages they capture are widely known.

% \paragraph{Representation.}
% The \emph{state diagram} of automaton $A$ is a directed graph with a vertex for
% each state $q \in Q$, with a distinguished vertex for $q_\mathrm{init}$, and a
% directed edge $q \to \delta(q,\sigma)$ for each $q \in Q$ and each $\sigma \in
% \Sigma$, labelled by the input/output pair $\langle \sigma, \theta(q,\sigma)
% \rangle$.

\begin{figure}[t]
  \centering
  \vspace{0.1cm}
  \includegraphics[draft=false,width=0.45\textwidth]{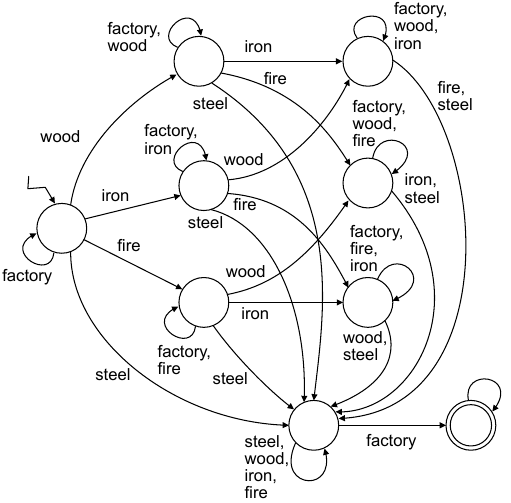}
  \caption{State diagram of the automaton for
  Example~\ref{ex:running-example-2}.}
  \label{figure:state-diagram}
\end{figure}

\begin{example}% [Running example, Part~\Roman{runningexamplecounter}]
  % \stepcounter{runningexamplecounter}
  \label{ex:running-example-2}
  The automaton for our running example reads traces generated by the
  agent while interacting with the environment.
  The input alphabet is 
  $\Sigma = \{ \mathrm{blank}$, 
  $\mathrm{wood},$ $\mathrm{iron},$ $\mathrm{fire},$ 
  $\mathrm{steel},$ $\mathrm{factory} \}$,
  where $\mathrm{blank}$ describes that no relvant event happened. 
  The automaton returns $1$ on traces where the agent has completed the task.
  The state diagram is depicted in Figure~\ref{figure:state-diagram},
  where transitions for $\mathrm{blank}$ are omitted and they always yield the
  current state. The output is $1$ on all transitions entering the double-lined
  state, and $0$ otherwise.
\end{example}

The automaton in the example has to keep track of the subset
of tasks completed so far, requiring one state per subset.
In general, the number of states can grow exponentially with the number of
tasks. 

\subsection{Existing Sample Complexity Results}

Classes of automata are typically defined in terms of the cardinality $k$ of
the input alphabet (assumed to be finite) and the number $n$ of states.
The existing result on the sample complexity of automata is for such a family of
classes.
\begin{theorem}[Ishigami and Tani, 1997]
  \label{th:ishigami-tani}
  Let $\mathcal{A}(k,n)$ be the class of $n$-state acceptors
  over the input alphabet $[1,k]$.
  Then,
  the sample complexity of\/ $\mathcal{A}(k,n)$ is 
  $\Theta(k \cdot n \cdot \log n)$.\footnote{The original result is equivalently
    stated in terms of the VC dimension, introduced later in this
  paper.}
\end{theorem}
Consequently, learning an acceptor from the class of all acceptors
with $k$ input letters and $n$ states requires an amount of data that is
at least $k \cdot n$.
Such a dependency is also observed in the existing automata learning algorithms,
e.g., 
\cite{angluin1987learning,ron1996amnesia,ron1998learnability,clark2004pac,palmer2007pac,balle2013pdfa,balle2014adaptively}.
More recently, there has been an effort in overcoming the direct dependency on
the cardinality $k$ of the input alphabet, through 
\emph{symbolic automata}
\cite{mens2015learning,maler2017generic,argyros2018symbolic}, but their sample
complexity has not been studied.

\section{Automata Cascades}
\label{sec:cascades}

We present the formalism of automata cascades, strongly based on the cascades
from \cite{krohn1965rhodes}---see also
\cite{ginzburg,maler1993thesis,domosi2005algebraic}.
The novelty of our formalism is that every cascade component is equipped with 
mechanisms for processing inputs and outputs. This allows for
(\emph{i}) controlling the complexity of components as the size of the cascade
increases; and 
(\emph{ii}) handling large (and even infinite) input alphabets.

% We do not claim any substantial novelty with respect to the
% original formalism. In fact, we aim at capturing the essence
% of the original cascades, while emphasising certain aspects such as controlling
% the increase of complexity of the components as the size of cascade increases.

\begin{definition}
An \emph{automata cascade} is a sequence of automata 
$A_1 \cascade \cdots \cascade A_d$ where
each $A_i$ is called a \emph{component} of the cascade and it is of the form 
$$\langle \Sigma_1 \times \cdots \times \Sigma_i, Q_i, \delta_i,
q^\mathrm{init}_i, \Sigma_{i+1}, \theta_i \rangle.$$
The function implemented by the cascade is the one implemented by the
automaton 
$\langle \Sigma_1, Q, \delta, q^\mathrm{init}, \Sigma_{d+1}, \theta
\rangle$
having
set of states
$Q = Q_1 \times \dots \times Q_d,$
initial state
$q_\mathrm{init} = \langle q_1^\mathrm{init}, \dots, q_d^\mathrm{init} \rangle$,
transition and output functions defined as
\begin{align*}
  \delta(\langle q_1, \dots, q_d \rangle, \sigma) & = \langle
\delta_1(q_1,\sigma_1), \dots, \delta_d(q_d,\sigma_d)\rangle,
\\
\theta(\langle q_1, \dots, q_d \rangle, \sigma) & = \theta_d(q_d, \sigma_d),
\end{align*}
where each component reads the recursively-defined input
$$
\sigma_1 = \sigma \text{, and }
\sigma_{i+1} = \langle \sigma_i, \theta_i(q_i,\sigma_i) \rangle.
$$
A cascade is \emph{simple} if $\theta_i(q,\sigma) = q$, for every 
$i \in [1,d-1]$.
An \emph{acceptor cascade} is a cascade where $\Gamma_{d+1} = \{0,1\}$.
\end{definition}
The components of a cascade are arranged in a sequence.
Every component reads the input and output of the preceding component, and
hence, recursively, it reads the external input together with the output of all
the preceding components.
%
% A cascade architecture is depicted in
% Figure~\ref{figure:fully-connected-cascade}.
The external input is the input to the first component, and the overall output
is the one of the last component.

As components of a cascade, we consider 
automata on factored alphabets that first apply a \emph{projection} operation on
their input, then apply a map to a smaller \emph{internal alphabet}, and finally
transition based on the result of the previous operations.
\begin{definition}
  An $n$-state \emph{automaton} is a tuple 
  $A = \langle X^a, J, \Pi, \phi, Q, \delta, q_\mathrm{init}, \Gamma, \theta
  \rangle$ where $X^a$ is the factored input alphabet;
  $J \subseteq [1,a]$ is the \emph{dependency set} and its cardinality $m$ is
  the \emph{degree of dependency};
  $\Pi$ is the finite \emph{internal alphabet};
  $\phi: X^m \to \Pi$ is an \emph{input function} that operates on the projected
  input tuples;
  $Q$ is a set of $n$ states;
  $\delta: Q \times \Pi \to Q$ is the transition function on the internal
  letters;
  $q_\mathrm{init} \in Q$ is the initial state;
  $\Gamma$ is the output alphabet; 
  and
  $\theta: Q \times X^m \to \Gamma$ is an output function that operates on the
  projected input tuples.
  The automaton induced by $A$ is the automaton
  $A' = 
  \langle \Sigma, Q, \delta_{J,\phi}, q_\mathrm{init}, \Gamma, \theta_J \rangle$
  where the input alphabet is $\Sigma = X^a$, the transition function is 
  $\delta_{J,\phi}(q,\sigma) =
  \delta(q,\phi(\pi_J(\sigma)))$, and the output function is 
  $\theta_J(q,\sigma) = \theta(q,\pi_J(\sigma))$.
  The core semiautomaton of $A$ is the core semiautomaton of $A'$.
  The string function implemented by $A$ is the one implemented by the induced
  automaton.
\end{definition}

The above definition adds two key aspects to the standard definition of
automaton.
First, the \emph{projection operation} $\pi_J$, that allows for capturing the
dependencies between the components in a cascade.
% , since input letters consist of outputs
% coming from the preceding components, together with the external input.
Although every component receives the output of all the preceding components,
it may use only some of them, and hence the others can be projected away. 
The dependency set $J$ corresponds to the indices of the input tuple that are
relevant to the component. 
Second, the \emph{input function} $\phi$, that maps the result of the projection
operation to an internal letter.
The rationale is that many inputs trigger the same transitions, and hence
they can be mapped to the same internal letter. 
This particularly allows for decoupling the size of the core
semiautomaton from the cardinality of the input alphabet---in line with
the mechanisms of symbolic automata
\cite{maler2017generic,argyros2018symbolic}.

\section{Expressivity of Automata Cascades}

Cascades can be built out of any set of components. However, the 
theory by Krohn and Rhodes identifies a set of \emph{prime automata} that is a
sufficient set of components to build cascades, as it allows for capturing all
automata. 
They are, in a sense, the building blocks of automata.
Moreover, using only some prime automata, we obtain specialised
expressivity results.

\subsection{Prime Components}

\begin{figure}[t]
  \centering
  \vspace{0.1cm}
  \includegraphics[draft=false,width=0.47\textwidth]{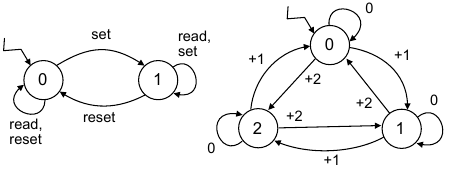}
  \caption{State diagrams of some of the simplest prime automata: a flip-flop on
    the left, and a $3$-counter on the right.}
  \label{figure:prime-automata}
\end{figure}

Prime automata are partitioned into two classes.
The \emph{first class of prime automata} are flip-flops, a kind of automaton that 
allows for storing one bit of information. 
\begin{definition}
A \emph{flip-flop} is a two-state automaton 
$\langle X^a, J, \Pi, \phi, Q, \delta, q_\mathrm{init}, \Gamma, \theta \rangle$
where $\Pi = \{ \mathrm{set}, \mathrm{reset}, \mathrm{read} \}$, 
$Q = \{ 0,1 \}$ and
the transition function satisfies the following three identities:
$$
\delta(q,\mathrm{read}) = q,\quad
\delta(q,\mathrm{set}) = 1,\quad
\delta(q,\mathrm{reset}) = 0.
$$
\end{definition}

We capture the task of our running example with a cascade where each
task is captured exactly by a flip-flop.

\begin{example}% [Running example, Part~\Roman{runningexamplecounter}]
  % \stepcounter{runningexamplecounter}
  \label{ex:running-example-cascade-simpler}
  The sequence task of our running example is captured by the cascade 
  $$
  A_\mathrm{wood} 
  \cascade 
  A_\mathrm{iron} 
  \cascade 
  A_\mathrm{fire} 
  \cascade 
  A_\mathrm{steel} 
  \cascade 
  A_\mathrm{factory} 
  $$
  where each component is a flip-flop that outputs its current state.
  The diagram for the cascade is shown in  
  Figure~\ref{figure:cascade-for-example}, where $\mathrm{getWood}$ corresponds
  to $A_\mathrm{wood}$, and similarly for the other components.
  All components read the input, and only $A_\mathrm{factory}$ also reads the
  output of the other components.
  % Note that, in this example, the external input alphabet is not factored.  
  Thus, 
  the dependency set of
  $A_\mathrm{wood},$ $A_\mathrm{iron},$ $A_\mathrm{fire},$ and
  $A_\mathrm{steel}$ is the singleton~$\{ 1 \}$, and the dependency set of
  $A_\mathrm{factory}$ is $\{ 1,2, 3, 4, 5 \}$---note that
  the indices correspond to positions of the components in the cascade.
  Then, $A_\mathrm{wood}$ has input function $\phi_\mathrm{wood}(x)$ that
  returns $\mathrm{set}$ if 
  $x = \mathrm{wood}$, and returns $\mathrm{read}$ otherwise.
  Similarly, $A_\mathrm{iron},$ $A_\mathrm{fire},$ and $A_\mathrm{steel}$.
  Instead, the component $A_\mathrm{factory}$ has 
  input function
  $\phi_\mathrm{factory}(x,\mathit{wood},\mathit{iron},\mathit{fire},\mathit{steel})$
  that returns $\mathrm{set}$ if 
  \begin{align*}
    (x = \mathrm{factory}) \land [(\mathit{wood} \land \mathit{iron} \land
  \mathit{fire}) \lor \mathit{steel}] 
  \end{align*}
  and returns $\mathrm{read}$ otherwise.
\end{example}

The \emph{second class of prime automata} is a class of automata that have a
correspondence with simple groups from group theory.
Their general definition is beyond the scope of this paper. For that, see
\cite{ginzburg}. Here we present the
class of prime counters, as a subclass that seems particularly
relevant from a practical point of view.
\begin{definition}
An \emph{$n$-counter} is an $n$-state automaton 
$\langle X^a, J, \Pi, \phi, Q, \delta, q_\mathrm{init}, \Gamma, \theta \rangle$
where $\Pi = Q = [0,n-1]$, and 
the transition function satisfies the following identity:
$$
\delta(i,j) = i+j\, (\operatorname{mod} n).
$$
An $n$-counter is \emph{prime} if $n$ is a prime number.
\end{definition}
An $n$-counter implements a counter modulo $n$. 
The internal letters correspond to numbers that allow for reaching any value of
the counter in one step. In particular, they also allow for implementing
the functionality of decreasing the counter, e.g., adding $n-1$ amounts to
subtracting $1$.
Note also that the internal letter $0$ plays the same role as $\mathit{read}$ does in
a flip-flop.
When we are interested just in counting---i.e., increment by one---the modulo
stands for the overflow due to finite memory.
On the other hand, we might actually be interested in counting modulo $n$; for
instance, to capture periodic events such as `something happens every 24 hours'.
A $3$-counter, that is a prime counter, is depicted in
Figure~\ref{figure:prime-automata}.

\begin{example}% [Running example, Part~\Roman{runningexamplecounter}]
  % \stepcounter{runningexamplecounter}
  \label{ex:counters}
  Resuming our running example, say that now, in order to use the factory, we
  need 
  \begin{romanenumerate*}
  \item
    $13$ pieces of wood, $5$ pieces of iron, and fire, or alternatively 
  \item 
    $7$ pieces of steel.
  \end{romanenumerate*}
  From the cascade in Example~\ref{ex:running-example-cascade-simpler}, it
  suffices to change the cascade components 
  $A_\mathrm{wood}, A_\mathrm{iron}, A_\mathrm{steel}$ into counters (e.g.,
  $16$-counters) and change the 
  input function of $A_\mathrm{factory}$ so that 
  $\phi_\mathrm{factory}(x,\mathit{wood},\mathit{iron},\mathit{fire},\mathit{steel})$
  returns $\mathrm{set}$ if 
  \begin{align*}
    & (x = \mathrm{factory}) 
    \land {}
    \\
    & \big[\big((\mathit{wood} \geq 13) \land (\mathit{iron}
    \geq 5) \land \mathit{fire}\big) \lor (\mathit{steel} \geq 7)\big],
  \end{align*}
  and it returns $\mathrm{read}$ otherwise.
  The rest is left unchanged.
  The cascade, despite its simplicity, corresponds to an automaton that has over
  $700$ states.
  Furthermore, suppose that we need to learn to detect $\mathrm{wood}$,
  $\mathrm{iron}$, $\mathrm{fire}$, $\mathrm{steel}$ from video frames
  represented as vectors of\/ $\mathbb{R}^a$.
  It suffices to replace the input function of the corresponding components with
  a function over\/ $\mathbb{R}^a$ such as a neural network.
\end{example}

\subsection{Expressivity Results}

A key aspect of automata cascades is their expressivity.
As an immediate consequence of the
Krohn-Rhodes theorem \cite{krohn1965rhodes}---see also 
\cite{ginzburg,maler1993thesis,domosi2005algebraic}---we have that simple
cascades of prime automata capture all automata, and simple cascades of
flip-flops capture the so-called \emph{group-free automata}.

% The following theorem is an immediate consequence of the
% Krohn-Rhodes Prime Decomposition
% theorem \cite{krohn1965rhodes}---see also 
% \cite{ginzburg,maler1993thesis,domosi2005algebraic}. 
% 
% It states that every semiautomaton is captured by its prime factors.
% However here we are not interested in decomposing a given automaton in terms of
% its factors---where the notion of factor requires some additional notions to be
% defined.
% Instead, we are interested in knowing what classes of components to consider in
% order to build cascades with a certain expressivity.
% Thus, we state the following (weaker) formulation of the theorem.

\begin{restatable}{theorem}{thkrohnrhodes}
  \label{th:krohn-rhodes}
  Every automaton is captured by a simple cascade of prime automata.
  Furthermore, 
  every group-free automaton is captured by a 
  simple cascade of flip-flops.
  The converse of both claims holds as well.
\end{restatable}
Group-free automata are important because they capture \emph{noncounting
automata}, c.f.\ \cite{ginzburg},
whose expressivity is the \emph{star-free regular languages}
\cite{schutzenberger1965finite}---i.e., the languages that can be specified by
regular expressions without using the Kleene star but using
complementation. This is the expressivity of well-known logics such as
\emph{monadic first-order logic on finite linearly-ordered domains}
\cite{mcnaughton1971counter}, and 
the \emph{linear temporal logic on finite traces} LTL$_f$ 
\cite{degiacomo2013ltlf}.
This observation, together with the fact that the expressivity of all acceptors
is the regular languages, allows us to derive the following theorem from the one
above.

% whose expressivity is the star-free regular languages 
% 
% Since every noncounting
% Defining group-free automata is beyond the scope of this paper, however we 
% 
% star-free is the expressivity of noncounting \cite{schutzenberger1965finite},
% and noncounting is group-free.
% Thus, by the theorem above, the expressivity of cascades of flip-flops is the
% star-free.
% 
% 
% The second part of the theorem is worth noting, since the expressivity of
% simple acceptor cascades of flip-flops  
% is the \emph{star-free regular languages} \cite{ginzburg}---i.e.,
% the languages that can be specified by regular expressions
% without using the Kleene star---which is
% the expressivity of well-known logics such as
% \emph{monadic first-order logic on finite linearly-ordered domains}
% \cite{mcnaughton1971counter}, and 
% the \emph{linear temporal logic on finite traces} LTL$_f$ 
% \cite{degiacomo2013ltlf}.

\begin{theorem}
  The expressivity of simple acceptor cascades of prime automata is the regular
  languages.
  The expressivity of simple acceptor cascades of flip-flops is the star-free
  regular languages.
\end{theorem}

\begin{figure}[t]
  \centering
  \vspace{0.1cm}
  \includegraphics[draft=false,width=0.45\textwidth]{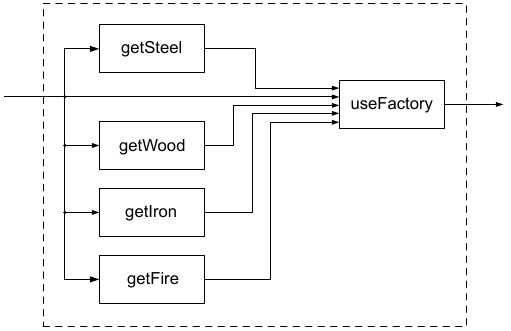}
  \caption{Diagram for the cascade of the running example.}
  \label{figure:cascade-for-example}
\end{figure}

\section{Sample Complexity of Automata Cascades}
\label{sec:sample-complexity}

We study the sample complexity of classes of automata (and cascades thereof)
built from given classes of input functions, semiautomata, and output functions. 
This allows for a fine-grained specification of automata and cascades. 
\begin{definition}
  The class $\mathcal{A}(\Phi,\Delta,\Theta)$ over a given input alphabet $X^a$ 
  consists of each automaton with input function from a given class $\Phi$, 
  core semiautomaton from a given class $\Delta$, 
  and output function from a give class $\Theta$.
  The classes $\Phi$, $\Delta$, $\Theta$ are uniform, and hence
  all automata in $\mathcal{A}(\Phi,\Delta,\Theta)$ have the 
  same degree of dependency $m$,
  same internal alphabet $\Pi$, 
  and 
  same output alphabet $\Gamma$;
  hence, sometimes we write $\mathcal{A}(\Phi,\Delta,\Theta; a, m)$ and
  $\mathcal{A}(\Phi,\Delta,\Theta; a, m, \Pi, \Gamma)$
  to have such quantities at hand.
\end{definition}

\subsection{Results for Finite Classes of Cascades}
\label{sec:sample-complexity-finite}

The following two theorems
establish upper bounds on the cardinality and sample complexity of finite
classes of automata, and finite classes of cascades, respectively.
The results are obtained by counting the number of ways in which we can
instantiate the parts of a single automaton. 
\begin{restatable}{theorem}{thfinitecomplexitygautomaton}
  \label{th:finite-complexity-gautomaton}
  The cardinality of a class of automata
  $\mathcal{A} = \mathcal{A}(\Phi,\Delta,\Theta; a, m)$ 
  is bounded as
  $$
  \left|\mathcal{A}\right| 
  \leq 
  \left|\pi_m^a\right| 
  \cdot 
  \left|\Phi\right| 
  \cdot 
  \left|\Delta\right| \cdot \left|\Theta\right|,
  $$
  and its sample complexity is asymptotically bounded as
  $$
  S(\mathcal{A}) \in 
    O\big(
        \log \left|\pi_m^a\right| 
        + 
        \log \left|\Phi\right| 
        + 
        \log \left|\Delta\right| 
        + 
        \log \left|\Theta\right|
        \big).
  $$
\end{restatable}
\begin{restatable}{theorem}{thfinitecomplexitycascade}
  \label{th:finite-complexity-cascade}
  The cardinality of a class of automata cascades
  $\mathcal{C} = 
  \mathcal{A}_1 \cascade \cdots \cascade \mathcal{A}_d$
  where the automata classes are 
  $\mathcal{A}_i =
  \mathcal{A}(\Phi_i,\Delta_i,\Theta_i;a_i,m_i)$
  is bounded as 
  $$
    \left|\mathcal{C}\right| 
    \leq 
    \prod_{i=1}^d 
    \left|\pi_{m_i}^{a_i}\right| 
    \cdot 
    \left|\Phi_i\right| 
    \cdot 
    \left|\Delta_i\right| 
    \cdot 
    \left|\Theta_i\right|,
  $$
  and its sample complexity is asymptotically bounded as
  $$
  S(\mathcal{C}) \in 
    O\Big(
        d \cdot (\log \left|\pi_m^a\right|
                 \,+\, 
                 \log \left|\Phi\right| 
                 \,+\, 
                 \log \left|\Delta\right| 
                 \,+\, 
                 \log \left|\Theta\right|) 
    \Big),
  $$
  where dropping the indices denotes the maximum.
  %% NOTE: Maximum quantities defined explicitly. 
  %%
  % $m = \max_i m_i$ is the maximum degree of dependency,
  % $a = \max_i a_i$ is the maximum arity of the input alphabet,
  % $\Phi = \argmax_{\Phi_i} |\Phi_i|$ is the class of input functions with
  % maximum cardinality, 
  % $\Delta = \argmax_{\Delta_i} |\Delta_i|$ is the class of semiautomata with
  % maximum cardinality,  
  % and
  % $\Theta = \argmax_{\Theta_i} |\Theta_i|$ is the class of output functions with
  % maximum cardinality.
\end{restatable}

Consequently, the complexity of a cascade is bounded by the product of the
number $d$ of components and a second factor that bounds the complexity of a
single component.

\subsection{Aspects and Implications of The Results}

The term $\log |\pi_m^a|$, accounting for the projection functions, ranges from
$0$ to $\min(m,a_1+d-m) \cdot \log (a_1+d)$ where $a_1$ is the input arity of
the first component, and hence of the external input.
In particular, it is minimum when we allow each component to depend on all or
none of its preceding components, and it is maximum when each component has to
choose half of its preceding components as its dependencies. 

The term $\log |\Phi|$ plays an important role, since the set of input functions
has to be sufficiently rich so as to map the external input and the outputs of
the preceding components to the internal input.
First, its cardinality can be controlled by the degree of dependency $m$; for
instance, taking all Boolean functions yields $\log |\Phi| = 2^m$. 
Notably, it does not depend on $d$, and
hence, the overall sample complexity grows linearly with the number of cascade
components as long as the degree of dependency $m$ is bounded.
Second, the class of input functions can be tailored towards the application at
hand. For instance, in our running example, input functions are chosen from a
class $\Phi$ for which $\log |\Phi|$ is linear in the number of tasks---see
Example~\ref{ex:running-example-complexity-finite} below.

The term $\log |\Delta|$ is the contribution of the \emph{number of
semiautomata}---note that the number of letters and states of each semiautomaton
plays no role.
Interestingly, very small classes of semiautomata suffice to build very
expressive cascades, by the results in Section~\ref{sec:cascades}.
% For instance, $\Delta$ can contain flip-flops only,
% which yields zero contribution to the sample complexity---see
% Example~\ref{ex:running-example-complexity-finite} below. 
In general, it is sufficient to build $\Delta$ out of prime components. 
But we can also include other semiautomata implementing some interesting
functionality.

The term $\log |\Theta|$ admits similar considerations as $\log |\Phi|$.
It is worth noting that one can focus on simple cascades,
by the results in Section~\ref{sec:cascades}, where output functions of all
but last component are fixed. 
Then, the contribution to the sample complexity is given by the class of output
functions of the last component. 

Overall, the sample complexity depends linearly on the number $d$ of components,
if we ignore logarithmic factors.
Specifically, the term $\log |\pi^a_m|$ has only a logarithmic dependency
on $d$, and the other terms $\log |\Phi|,$  
$\log |\Delta|,$  and $\log |\Theta|$ are independent of $d$. 

\begin{corollary}
  Let us recall the quantities from Theorem~\ref{th:finite-complexity-cascade},
  and fix the quantities $a$, $m$, $\Phi$, $\Delta$, and $\Theta$.
  Then, the sample complexity of $\mathcal{C}$ is bounded as
  $S(\mathcal{C}) \in O(d \cdot \log d)$.
\end{corollary}

In the next example we instantiate a class of cascades for our
running example, and derive its sample complexity.

\begin{example}% [Running example, Part~\Roman{runningexamplecounter}]
  % \stepcounter{runningexamplecounter}
  %
  \label{ex:running-example-complexity-finite}
  The cascade described 
  in Example~\ref{ex:running-example-cascade-simpler} has one component per
  task, and all components have the same output function and
  semiautomaton.
  The input function $\phi_\mathrm{factory}$ is $2$-term monotone DNF over $9$
  propositional variables. Every other component has an input function that is 
  $1$-term monotone DNF over $5$ propositional variables.
  Using these observations, we can design a class of cascades for similar
  sequence tasks where the goal task depends on two groups of arbitrary size,
  having $d$ basic tasks overall.
  Such a class will consist of cascades of $d$ components.
  For every $i \in [1,d-1]$,
  the class of input functions $\Phi_i$ is the class of $1$-term monotone DNF
  over $d$ variables;
  and $\Phi_d$ is the class of $2$-term monotone DNF over $2d-1$ variables.
  For all $i \in [1,d]$, 
  $\Delta_i$ is a singleton consisting of a flip-flop semiautomaton;
  and $\Theta_i$ is a singleton consisting of the function that returns the
  state.
  The cardinality of $\Phi_i$ is $e^2 \cdot 2^{(4d-4)}$, and hence
  its logarithm is less than $4 d$---see 
  Corollary~5\ of \cite{schmitt2004sparse}.
  By Theorem~\ref{th:finite-complexity-cascade}, 
  considering that we have $d$ cascade components, 
  we obtain that our class of cascades
  has sample complexity $O(d^2)$.
  At the same time, the minimum automaton for the considered family of sequence
  tasks has $\Omega(2^d)$ states. 
  Going back to the bound of Theorem~\ref{th:ishigami-tani},
  if we had to learn from the class of all automata with $2^d$ states, we would
  incur a sample complexity exponential in the number of tasks.
\end{example}

\subsection{Learning Theory for Infinite Classes}
\label{sec:advanced-learning-theory}
Infinite classes of cascades naturally arise when the input alphabet is 
infinite. In this case, one may consider to pick input and output functions
from an infinite class. For instance, the class of all threshold functions over 
the domain of integers, or as in Example~\ref{ex:counters} the neural
networks over the vectors of real numbers. 

When considering infinite classes of functions, the sample complexity bounds
based on cardinality become trivial. 
However, even when the class is infinite, the number of functions from the class
that can be distinguished on a given sample is finite, and the way it grows as a
function of the sample size allows for establishing sample complexity bounds.
In turn, such \emph{growth} can be bounded in terms of the \emph{dimension} of
the class of functions, a single-number characterisation of its complexity.
Next we present these notions formally. 

\paragraph{Growth and Dimension.}
Let $\mathcal{F}$ be a class of functions from a set $X$ to a finite
set $Y$.
Let $X_\ell = x_1, \dots, x_\ell$ be a sequence of $\ell$ elements from $X$.
The set of \emph{patterns} of a class $\mathcal{F}$ on $X_\ell$ is 
    $$\mathcal{F}(X_\ell) = \{ \langle f(x_1), \dots, f(x_\ell) \rangle \mid f \in
    \mathcal{F} \},$$
and the \emph{number of distinct patterns} of class $\mathcal{F}$ on $X_\ell$ 
is 
$N(\mathcal{F},X_\ell) = |\mathcal{F}(X_\ell)|$.
The \emph{growth} of $\mathcal{F}$ is
$$
G(\mathcal{F},\ell) = \operatorname{sup}_{X_\ell} N(\mathcal{F}, X_\ell).
$$
The growth of $\mathcal{F}$ can be bounded using its \emph{dimension},
written as $\dimension(\mathcal{F})$. 
When $|Y| = 2$, we define the dimension of $\mathcal{F}$ to be its 
\emph{VC dimension} \cite{vapnik1971uniform}---see also
\cite{vapnik1998statistical}. It is the largest integer
$h$ such that $G(\mathcal{F},h) = 2^{h}$ and $G(\mathcal{F},h+1) < 2^{h+1}$ if
such an $h$ exists, and infinity otherwise.
When $|Y| > 2$, we define the dimension of $\mathcal{F}$ to be its \emph{graph
dimension} \cite{natarajan1989learning,haussler1995generalization}. 
It is defined by first binarising the class of functions. 
For a given function $f: X \to Y$, its binarisation 
$f_\mathrm{bin}: X \times Y \to \{ 0,1 \}$ is defined as 
$f_\mathrm{bin}(x,y) = \mathbf{1}[f(x) = y]$.
The \emph{binarisation} of $\mathcal{F}$ is 
$\mathcal{F}_\mathrm{bin} = \{ f_\mathrm{bin} \mid f \in \mathcal{F} \}$.
Then, the graph dimension of $\mathcal{F}$ is defined as the VC dimension
of its binarisation $\mathcal{F}_\mathrm{bin}$.

The growth of $\mathcal{F}$ can be bounded in terms of
its dimension \cite{haussler1995generalization}, as follows:
$$G(\mathcal{F},\ell) \leq 
\left(e \cdot \ell \cdot |Y| \right)^{\dimension(\mathcal{F})}.$$

\paragraph{Sample Complexity.}
The sample complexity can be bounded in terms of the dimension, cf.\ 
\cite{shalevshwartz2014book}. In particular,
\begin{equation*}
S(\mathcal{F},\epsilon,\eta) \in 
O\left((\dimension(\mathcal{F}) \cdot \log|Y| - \log \eta)/\epsilon^2\right).
\end{equation*}
For fixed $\epsilon$, $\eta$, and $Y$,
the sample complexity is $O(\dimension(\mathcal{F}))$, and hence
arbitrary classes over the same outputs can be compared in terms of their 
dimension.

\subsection{Results for Infinite Classes of Cascades}
\label{sec:sample-complexity-infinite}

We generalise our sample complexity bounds to infinite classes of automata and
cascades.
The bound have the same shape of the bounds derived in
Section~\ref{sec:sample-complexity-finite} for finite classes, with the 
dimension replacing the logarithm of the cardinality.
One notable difference is the logarithmic dependency on the maximum length $M$
of a string. It occurs due to the (stateless) input and output functions. 
Their growth is on single letters, regardless of the way they are grouped into
strings. Thus, the growth on a sample of $\ell$ strings is the growth on a sample
of $\ell \cdot M$ letters.

The bounds are derived using a functional description of automata and cascades.
It is a shift of perspective, from stateful machines that process one letter at
a time, to black-boxes that process an entire input string at once.
We first introduce function constructors to help us to build such descriptions.

\begin{definition}[Function constructors]
  \label{def:function-constructors}
  For $f: \Sigma \to \Gamma$,
  $f^*: \Sigma^* \to \Gamma$ is defined as 
  $f^*(\sigma_1 \dots \sigma_n) = f(\sigma_n)$.
  For $g: \Sigma^* \to \Gamma$,
  $\overline{g}: \Sigma^* \to \Gamma^*$ is defined as 
  $\overline{g}(\sigma_1 \dots \sigma_n) = g(\sigma_1) \dots g(\sigma_1 \dots
  \sigma_n)$;
  furthermore,
  $g^\triangleleft: \Sigma^* \to \Gamma^*$ is defined as 
  $g^\triangleleft(\sigma_1 \dots \sigma_n) = g(\sigma_1 \dots \sigma_{n-1})$.
\end{definition}

\begin{restatable}{lemma}{lemmafunctioncompositiongeneralisedautomaton}
  \label{lemma:function-composition-generalised-automaton}
  The function $A$ implemented by an automaton 
  $\langle X^a, J, \Pi, \phi, Q, \delta, q_\mathrm{init}, \Gamma, \theta
  \rangle$
  can be expressed as
  \begin{equation*}
    A = 
    \overline{\pi^*_J} \circ ((\overline{\phi^*} \circ D^\popfunc)
    \times I^*) \circ \theta,
  \end{equation*}
  where $D$ is the function implemented by the semiautomaton
  $\langle \Pi, Q, \delta, q_\mathrm{init} \rangle$.
\end{restatable}

Note that $I^*$ propagates the projected input to the output function.
Also note that the output function reads the state before the last update.
The functional description allows us to bound the growth, by making use of the
fact that the growth of composition and cross product of two classes of
functions is upper bounded by the product of their respective growths.
From there, we derive the dimension, and the sample complexity.
\begin{theorem}
  \label{theorem:vc-generalised-automata}
  Let $\mathcal{A}$ be a class 
  $\mathcal{A}(\Phi,\Delta,\Theta;a,m,\Pi,\Gamma)$, 
  let $M$ be the maximum length of a string, and
  let $w = \log\left|\pi_m^a\right| + \log\left|\Delta\right| +
  \dimension(\Phi) + \dimension(\Theta) \geq 2$.
  \smallskip\par\noindent
  (i) The growth of $\mathcal{A}$ is bounded as:
  $$
    G(\mathcal{A},\ell) 
    \leq 
    \left|\pi^a_m\right|
    \cdot 
    \left|\Delta \right|
    \cdot 
    G(\Phi,\ell \cdot M) 
    \cdot 
    G(\Theta,\ell).
  $$
  (ii) The dimension of $\mathcal{A}$ is bounded as:
  $$
    \dimension(\mathcal{A}) 
    \leq 2 \cdot w \cdot \log(w \cdot e \cdot M \cdot |\Pi| \cdot |\Gamma|).
  $$
  (iii) The sample complexity of $\mathcal{A}$ is
  bounded as:
  $$
    S(\mathcal{A}) \in 
    O\big(w \cdot \log(w \cdot M \cdot |\Pi| \cdot |\Gamma|)\big).
  $$
\end{theorem}

Next we show the functional description of an automata cascade.
It captures the high-level idea that each component of a cascade can be
fully executed before executing any of the subsequent automata, since it does
not depend on them.
\begin{restatable}{lemma}{lemmafunctioncompoisitionofcascade}
  \label{lemma:function-compoisition-of-cascade}
  The function implemented by an automata cascade 
  $A_1 \cascade \cdots \cascade A_d$ with $d \geq 2$ can be expressed as
  \begin{equation*}
    \overline{(I^* \times A_1)} \circ \dots \circ
    \overline{(I^* \times A_{d-1})} \circ A_d.
  \end{equation*}
\end{restatable}

The cross product with $I^*$ in the functional description above captures
the fact that the input of each component is propagated to the next one.

Then, a bound on the growth is derived from the functional description, and
hence the dimension and sample complexity bounds.
\begin{theorem}
  \label{theorem:v-dimension-cascades}
  Let $\mathcal{C}$ be a class 
  $\mathcal{A}_1 \cascade \dots \cascade \mathcal{A}_d$ 
  where automata classes are
  $
  \mathcal{A}_i = 
  \mathcal{A}(\Phi_i,\Delta_i,\Theta_i;a_i,m_i,\Pi_i,\Gamma_i)
  $,
  let $M$ be the maximum length of a string, and 
  let $w = \log\left|\pi_m^a\right| + \log\left|\Delta\right| +
  \dimension(\Phi) + \dimension(\Theta) \geq 2$ where dropping the indices
  denotes the maximum.
  \smallskip\par\noindent
  (i) The growth of\/ $\mathcal{C}$ is bounded as:
  $$
    G(\mathcal{C}, \ell) 
    \leq 
    \prod_{i=1}^d 
    \left|\pi_{m_i}^{a_i}\right|
    \cdot 
    \left|\Delta_i \right| 
    \cdot
    G(\Phi_i,\ell \cdot M) 
    \cdot
    G(\Theta_i,\ell \cdot M).
  $$
  (ii) The dimension of\/ $\mathcal{C}$ is bounded as:
  $$
    \dimension(\mathcal{C}) 
    \leq 
    2 \cdot d \cdot w 
    \cdot \log(d \cdot w \cdot e \cdot M \cdot |\Pi| \cdot |\Gamma|).
  $$
  (iii) The sample complexity of\/ $\mathcal{C}$ is bounded as:
  $$
    S(\mathcal{C}) 
    \in 
    O\big(
      d \cdot w \cdot \log(d \cdot w \cdot M \cdot |\Pi| \cdot |\Gamma|)
    \big).
  $$
\end{theorem}

By a similar reasoning as in the results for finite classes, the sample
complexity has a linear dependency on the number $d$ of cascade components,
modulo a logarithmic factor.

\begin{corollary}
  Let us recall the quantities from Theorem~\ref{theorem:v-dimension-cascades},
  and fix the quantities $a_1$, $m$, $\Phi$, $\Delta$, $\Theta$, $M$, $\Pi$, and
  $\Gamma$.
  Then, the sample complexity of $\mathcal{C}$ is bounded as
  $S(\mathcal{C}) \in O(d \cdot \log d)$.
\end{corollary}

\section{Related Work}

The main related work is the sample complexity bounds from
\cite{ishigami1997vc}, stated in Theorem~\ref{th:ishigami-tani}.
Our bounds are qualitatively different, as they allow for describing the sample
complexity of a richer variety of classes of automata, and cascades thereof.
With reference to Theorem~\ref{th:finite-complexity-gautomaton},
their specific case is obtained when:
the input alphabet is non-factored and hence $|\pi^a_m| = 1$;
$\Phi = \{ I \}$ and hence $|\Pi| = |\Sigma| = k$;
$\Delta$ is the class of all semiautomata on $k$ letters and $n$ states, and
hence $|\Delta| = n^{k \cdot n}$; 
and $\Theta$ is the class of all indicator functions on $n$ states, and hence
$|\Theta| = 2^n$.
In this case, it is easy to verify that our bound matches theirs.

Learning automata expressed as a cross product
is considered in \cite{moerman2018product}. They correspond to cascades where
all components read the input, but no component reads the output of the others. 
The authors provide a so-called \emph{active learning} algorithm, that asks
membership and equivalence queries. Although it is a different setting from
ours, it is interesting that they observe an exponential gain in some specific
cases, compared to ignoring the product structure of the automata.

The idea of decoupling the input alphabet from the core functioning of an
automaton is found in \emph{symbolic automata}.
The existing results on learning symbolic automata are in the active
learning setting
\cite{berg2006regular,mens2015learning,maler2017generic,argyros2018symbolic}.

% \ifarxiv

\section{Conclusion}

Given the favourable sample complexity of automata cascades, the next step is to
devise learning algorithms, able to learn automata as complex systems consisting
of many components implementing specific functionalities.

\section*{Acknowledgments}
This work has been supported by the ERC Advanced Grant WhiteMech (No.\
834228), by the EU ICT-48 2020 project TAILOR (No.\ 952215), by the PRIN project
RIPER (No.\ 20203FFYLK), by the EU's Horizon 2020 research and
innovation programme under grant agreement No.\ 682588.

%%% Rerefences

% Smaller font size for references in case of need
% \fontsize{9.8pt}{10.8pt} \selectfont

% \bibliography{journal-abbreviations,bibliography,bibliography-0}
\bibliography{extracted-bibliography}

\ifarxiv
  \onecolumn
  \clearpage
  \appendix
  
\section{Proofs}

\subsection{Proof of Theorem~\ref{th:krohn-rhodes} (Expressivity of Automata
Cascades)}

% We will need the definition of homomorphism.
% \begin{definition}
% A \emph{homomorphism} from a semiautomaton 
% $A_1 = \langle \Sigma, Q_1, \delta_1, q_1^\mathrm{init} \rangle$ to a semiautomaton 
% $A_2 = \langle \Sigma, Q_2, \delta_2, q_2^\mathrm{init} \rangle$
% is a surjective function $\psi: Q_1 \to Q_2$ such that the two identities
% $\psi(q^\mathrm{init}_1) = q^\mathrm{init}_2$ and
% $\psi(\delta_1(q,\sigma)) = \delta_2(\psi(q),\sigma)$ are valid.
% In such a case we say that $A_2$ is homomorphic to $A_1$, and we
% write $A_2 \preceq A_1$.
% \end{definition}

We first introduce the \emph{cascade product} between semiautomata.
Given two semiautomata 
\begin{align*}
  D_1 & = \langle \Sigma, Q, \delta_1, q_\mathrm{init}^1 \rangle,
  \\
  D_2 & = \langle \Sigma \times Q, Q, \delta_2, q_\mathrm{init}^2 \rangle,
\end{align*}
their cascade product $D_1 \cascade D_2$ yields the semiautomaton
$D = \langle \Sigma, Q, \delta, q_\mathrm{init} \rangle$
where 
\begin{align*}
  Q & = Q_1 \times Q_2,
  \\
  q_\mathrm{init} & = \langle q_\mathrm{init}^1, q_\mathrm{init}^2 \rangle,
  \\
  \delta(\langle q_1, q_2 \rangle, \sigma) & = \langle \delta_1(q_1,
\sigma), \delta_2(q_2, \langle q_1, \sigma \rangle) \rangle.
\end{align*}

The cascade product on semiautomata is left-associative. Hence, given
semiautomata 
$D_i = \langle \Sigma \times Q_1 \times \dots \times Q_{i-1}, Q_i, \delta_i,
q_\mathrm{init}^i \rangle$ for $i \in [1,d]$,
we can build a semiautomaton as $D_1 \cascade \cdots \cascade D_d$.

The proof of Theorem~\ref{th:krohn-rhodes} is based on the following
theorem, which is an immediate consequence of the Krohn-Rhodes Prime
Decomposition Theorem \cite{krohn1965rhodes}.

\begin{theorem}[Krohn and Rhodes]
  \label{th:krg}
  Every automaton $A$ is captured by an automaton $A'$ whose core semiautomaton
  is expressed by a cascade product $A_1 \cascade \cdots \cascade A_d$ where
  each $A_i$ is a prime semiautomaton.  
  Furthermore,
  every group-free automaton $A$ is captured by an automaton $A'$ whose core
  semiautomaton is expressed by a cascade product $A_1 \cascade \cdots \cascade
  A_d$ where each $A_i$ is a flip-flop semiautomaton.
  The converse of both claims holds as well.
\end{theorem}

% \begin{theorem}[Theorem 5.6 of \cite{ginzburg}]
%   Let $A$ be an automaton with core semiautomaton $A'$, and let 
%   $B'$ be a semiautomaton such that $A' \preceq B'$. Then, there exists an
%   automaton $B$, with core semiautomaton $B'$, that captures $A$.
% \end{theorem}
% 
% \begin{lemma}
%   Every cascade product of prime (flip-flop) semiautomata is captured by a
%   simple cascade of prime (flip-flop) automata.
%   The converse holds as well.
% \end{lemma}
% \begin{proof}
%   TODO
% \end{proof}
% 
% Then, in the following proof, we use Theorem~5.6 of \cite{ginzburg}, saying that
% every 
%
\thkrohnrhodes*
\begin{proof}
  Consider an automaton 
  $A = \langle \Sigma, Q, \delta, q_\mathrm{init}, \Gamma, \theta \rangle$.
  By Theorem~\ref{th:krg},
  automaton
  $A$ is captured by an automaton
  $A' = \langle \Sigma, Q', \delta', q'_\mathrm{init}, \Gamma,
  \theta' \rangle$
  such that its semiautomaton $\langle \Sigma, Q', \delta', q_\mathrm{init}'
  \rangle$ is expressed by
  a cascade product
  $D_1 \cascade \cdots \cascade D_d$ where each $D_i$ is a prime semiautomaton
  of the form
  $\langle \Sigma \times Q_1 \times \dots \times Q_{i-1}, Q_i, \delta_i,
  q_\mathrm{init}^i \rangle$.
  For every $i \in [1,d-1]$, let $A_i$ be the automaton with core semiautomaton 
  $D_i$ and output function 
  $\theta_i(q_i,\langle \sigma, q_1, \dots, q_{i-1} \rangle) = q_i$.
  Also, let $A_d$ be the automaton with core semiautomaton $D_d$ and 
  output function
  $\theta_d(q_d,\langle \sigma, q_1, \dots, q_{d-1} \rangle) =
  \theta'(\langle q_1, \dots, q_d \rangle,\sigma)$.
  Now, consider the cascade $C = (A_1, \dots, A_d)$ made of the automata just
  introduced.
  Since the cascade is simple and its components are prime,
  it suffices to show that $A$ is captured by $C$.
  Namely, that $A$ and $C$ implement the same function.
  Let 
  $A_C = \langle \Sigma, Q_C, \delta_C, q^\mathrm{init}_C, \Gamma, \theta_C
  \rangle$ 
  be the automaton corresponding to the cascade $C$.
  Consider an input string $\sigma_1 \dots \sigma_m$.
  Let $q = \langle q_1, \dots, q_d \rangle$ be the state of $A_C$ after
  reading $\sigma_1 \dots \sigma_{m-1}$; namely, 
  $q = \delta_C(q^\mathrm{init}_C, \sigma_1 \dots \sigma_{m-1})$.
  The function implemented by $C$ is
  \begin{align*}
    C(\sigma_1 \dots \sigma_m) = \theta_d(q_d, \langle \sigma_m, q_1, \dots,
    q_{d-1} \rangle)
    = \theta'(\langle q_1, \dots, q_d \rangle,\sigma_m) = 
    A'(\sigma_1 \dots \sigma_m) = A(\sigma_1 \dots \sigma_m).
  \end{align*}
  For the second claim, it suffices to note that $A_1, \dots, A_d$ can be taken
  to be flip-flops when $A$ is group-free, by Theorem~\ref{th:krg}.

  Next we show the converse of the two claims.
  The converse of the first claim holds immediately by the definition of
  cascade, which says that the function it implements is the one implemented by
  the corresponding automaton.  
  For the converse of the second claim, consider a simple cascade 
  $C = (A_1, \dots, A_d)$. For every $i \in [1,d]$, the core semiautomaton $D_i$
  of $A_i$ is a flip-flop. Also, for every $i \in [1,d-1]$ the output function 
  of $A_i$ is $\theta_i(q_i,\langle \sigma, q_1, \dots, q_{i-1} \rangle) = q_i$.
  Thus, the core semiautomaton of the automaton corresponding to $C$ is
  expressed by the cascade product $D_1 \cascade \cdots \cascade D_d$.
  Thus, by Theorem~\ref{th:krg}, cascade $C$ is captured by a 
  group-free automaton.
\end{proof}

\subsection{Proof of Theorem~\ref{th:finite-complexity-gautomaton} (Cardinality
  and Sample Complexity of Finite Classes of Automata)}

\thfinitecomplexitygautomaton*
\begin{proof}
  We bound the number of automata in the class
  $\mathcal{A}(m,\Phi,\Delta,\Theta)$.
  Consider that each automaton in the class is of the form
  $$
    A = \langle 
        X^a, J, \Pi, \phi, Q, \delta, q_\mathrm{init}, \Gamma, \theta
    \rangle.
  $$
  To establish the cardinality of the class, it suffices to count the number of
  distinct components of the tuple above.
  The input alphabet $X^a$ is fixed. 
  The number of possible sets of indices of $J \subseteq [1,a]$ is the 
  cardinality of
  the set $\left|\pi_m^a\right| = \left| \binom{a}{m}\right|$.
  The number of input functions $\phi$ is $|\Phi|$.
  Each input function determines the internal alphabet $\Pi$.
  Each semiautomaton $\langle \Pi, Q, \delta, q_\mathrm{init} \rangle$ is from
  $\Delta$, and hence we have at most $|\Delta|$.
  The number of output functions $\theta$ is $|\Theta|$, and
  the output alphabet $\Gamma$ is determined by the output function.
  The bound follows by taking the product of the sets of components whose
  cardinality has been discussed above.

  For the case where $\Delta$ is the set of all semiautomaton with $k$ letters
  and $n$ states, we have $|\Delta| \leq n^{k \cdot n}$---assuming w.l.o.g.\
  that $Q = [1,n]$ and $q_\mathrm{init} = 1$.

  The sample complexity is immediate by taking the logarithm of the cardinality
  bound established above.
\end{proof}

\subsection{Proof of Theorem~\ref{th:finite-complexity-cascade} (Cardinality and
Sample Complexity of Finite Classes of Automata Cascades)}

Note that the theorem is restated by defining some of the quantities more
clearly.

\medskip
\par
\noindent
\textbf{Theorem~\ref{th:finite-complexity-cascade}.}
\textit{%
The cardinality of a class of automata cascades
  $\mathcal{C} = 
  \mathcal{A}_1 \cascade \cdots \cascade \mathcal{A}_d$
  where the automata are from
  $\mathcal{A}_i =
  \mathcal{A}(\Phi_i,\Delta_i,\Theta_i;a_i,m_i)$
  is bounded as 
  $$
    \left|\mathcal{C}\right| 
    \leq 
    \prod_{i=1}^d 
    \left|\pi_{m_i}^{a_i}\right| 
    \cdot 
    \left|\Phi_i\right| 
    \cdot 
    \left|\Delta_i\right| 
    \cdot 
    \left|\Theta_i\right|,
  $$
  and its sample complexity is asymptotically bounded as
  $$
  S(\mathcal{C}) \in 
    O\Big(
        d \cdot (\log \left|\pi_m^a\right|
                 \,+\, 
                 \log \left|\Phi\right| 
                 \,+\, 
                 \log \left|\Delta\right| 
                 \,+\, 
                 \log \left|\Theta\right|) 
    \Big),
  $$
  where dropping the indices denotes the maximum as follows:
  \begin{alignat*}{2}
    \pi_m^a & = \argmax_{\pi_{m_i}^{a_i}} |\pi_{m_i}^{a_i}| 
            && \qquad\mbox{the class of projection functions with maximum cardinality,}
    \\
    \Phi & = \argmax_{\Phi_i} |\Phi_i|
         && \qquad\mbox{the class of input functions with maximum cardinality,}
    \\
    \Delta & = \argmax_{\Delta_i} |\Delta_i|
           && \qquad\mbox{the class of semiautomata with maximum cardinality,}
    \\
    \Theta & = \argmax_{\Theta_i} |\Theta_i|
           && \qquad\mbox{the class of output functions with maximum cardinality.}
  \end{alignat*}
}
\begin{proof}
  We bound the number of cascades in the class 
  $\mathcal{C} = \mathcal{A}_1 \cascade \cdots \cascade \mathcal{A}_d$.
  In order to do so, we bound the cardinality of each class of automata
  $\mathcal{A}_i = \mathcal{A}(\Phi_i,\Delta_i,\Theta_i;a_i,m_i)$.
  Thus, by the cardinality bound in
  Theorem~\ref{th:finite-complexity-gautomaton}, we have that 
  $$
      \left|\mathcal{A}_i\right| 
      \leq 
      \left|\pi_{m_i}^{a_i}\right| 
      \cdot 
      \left|\Phi_i\right| 
      \cdot 
      \left|\Delta_i\right| 
      \cdot 
      \left|\Theta_i\right|.
  $$
  Thus, the bound on $\left|\mathcal{C}\right|$ is given by 
  $\prod_{i=1}^d\left|\mathcal{A}_i\right|$. 
  Then, for the sample complexity, we loosen the bound 
  taking the maximum of each indexed quantity.  Namely,
  \begin{align*}
    \left|\mathcal{C}\right| 
    & \leq 
    \prod_{i=1}^d \left|\mathcal{A}_i\right| 
    \leq
    \prod_{i=1}^d 
    \left|\pi_m^a\right|
    \cdot 
    \left|\Phi\right| 
    \cdot 
    \left|\Delta\right| 
    \cdot 
    \left|\Theta\right|,
  \end{align*}
  where $\pi_m^a,$ $\Phi,$ $\Delta,$ and $\Theta$ 
  are defined in the statement of the theorem.
  Then,
  $$
      \left|\mathcal{C}\right| 
      \leq 
      \left|\pi_m^a\right|^d
      \cdot 
      \left|\Phi\right|^d 
      \cdot 
      \left|\Delta\right|^d 
      \cdot
      \left|\Theta\right|^d.
  $$
  The sample complexity bound follows by taking the logarithm of the
  cardinality bound established above.
\end{proof}

\subsection{Further Details on the Growth and Sample Complexity}

This section restates two claims from
Section~\ref{sec:advanced-learning-theory}, providing details on the way they
are derived.
The following Proposition~\ref{prop:bound-growth-vc-haussler} restates
the upper bound on the growth in terms of the dimension.
\begin{proposition}
  \label{prop:bound-growth-vc-haussler}
  Let $\mathcal{F}$ be a class of functions from $X$ to $Y$.
  Then, 
  $$G(\mathcal{F},\ell) 
  \;\leq \;
  \sum_{i=0}^{\dimension(\mathcal{F})} \binom{\ell}{i} \cdot (|Y|-1)^i
  \;\leq \;
  \sum_{i=0}^{\dimension(\mathcal{F})} \binom{\ell}{i} \cdot |Y|^i
  \;\leq\;
  \left(e \cdot \ell \cdot |Y|\right)^{\dimension(\mathcal{F})}.$$
\end{proposition}
\begin{proof}
  The first inequality is Corollary~3 of \cite{haussler1995generalization}. We 
  loosen it by  substituting $(|Y|-1)^i$ with $|Y|^i$ and then
  applying the well-known bound on the partial sum of binomials 
  $\sum_{i=0}^d \binom{m}{i} \leq (e \cdot m)^d$ to derive the last inequality.
\end{proof}

The following Proposition~\ref{prop:sample-complexity-derivation} restates the
asymptotic bound on the sample complexity in terms of the dimension.
\begin{proposition}
  \label{prop:sample-complexity-derivation}
  The following asymptotic bound on the sample complexity holds true:
\begin{equation*}
S(\mathcal{F},\epsilon,\eta) \in 
O\left((\dimension(\mathcal{F}) \cdot \log|Y| - \log \eta)/\epsilon^2\right).
\end{equation*}
\end{proposition}
\begin{proof}
  By Point~2 of Theorem~19.3 of \cite{shalevshwartz2014book}, 
  the proposition holds with the \emph{Natarajan dimension} of
  $\mathcal{F}$ in place of $\dimension(\mathcal{F})$.
  Then the proposition follows immediately since 
  $\dimension(\mathcal{F})$ is an upper bound on the Natarajan dimension, 
  by Equations~(24) and~(26) of \cite{haussler1995generalization}.
\end{proof}

\subsection{Basic Propositions}

In this section we state and prove some basic propositions that we use 
in the next sections of the appendix, where we prove the
results for Section~\ref{sec:sample-complexity-infinite}.

The next Proposition~\ref{prop:inequality-two-growths} describes a relationship
between the growth of $\mathcal{F}$ and the growth of its binarisation
$\mathcal{F}_\mathrm{bin}$. 
\begin{proposition}
  \label{prop:inequality-two-growths}
Let $\mathcal{F}$ be a class of functions. 
Then, $G(\mathcal{F}_\mathrm{bin},\ell) \leq G(\mathcal{F},\ell)$.
\end{proposition}
\begin{proof}
  Let $Z_\ell = \langle x_1, y_1 \rangle, \dots, \langle x_\ell, y_\ell \rangle$
  and let $X_\ell = x_1, \dots, x_\ell$.
  The proposition is immediate by observing that
  every tuple in 
  $$
  \mathcal{F}(X_\ell) = \{ \langle f(x_1), \dots, f(x_\ell) \rangle \mid f \in
  \mathcal{F} \}
  $$
  corresponds to at most one tuple in
  $$
  \mathcal{F}_\mathrm{bin}(Z_\ell) =
  \{  
      \langle \mathbf{1}[y_1 \!=\! f(x_1)], \dots, \mathbf{1}[y_\ell \!=\!
        f(x_\ell)] \rangle 
  \mid f \in \mathcal{F} \}.
  $$
\end{proof}

The next Propositions~\ref{prop:growth-composition}
and~\ref{prop:growth-parallel-composition} provide bounds on the growth for
composition and cross product of functions. 
They are well-known. 
For instance, they are left as an exercise in the book
\cite{shalevshwartz2014book}. 
However, we were not able to find a peer-reviewed reference for the proofs.
We report proofs based on the lecture notes (Kakade and Tewari 2008).

%
% Proposition: growth of composition is a product of growths
%
\begin{restatable}{proposition}{propgrowthcomposition}
\label{prop:growth-composition}
  Let $\mathcal{F}_1: X \to W$, and let 
      $\mathcal{F}_2: W \to Y$.
  Then,
  $
    G(\mathcal{F}_1 \circ \mathcal{F}_2,\ell) 
    \leq 
    G(\mathcal{F}_1,\ell) \cdot G(\mathcal{F}_2,\ell).
  $
\end{restatable}
\begin{proof}
  Let $X_\ell = x_1, \dots, x_\ell$ be a sequence of $\ell$ elements from
  $X$.
  Let $\mathcal{F} = \mathcal{F}_1 \circ \mathcal{F}_2$.
  We have
  \begin{align*}
    \mathcal{F}(X_\ell) 
    & = 
    \{ \langle f_2(f_1(x_1)), \dots, f_2(f_1(x_\ell))
  \rangle \mid f_1 \in \mathcal{F}_1\setcomma f_2 \in \mathcal{F}_2 \}
  \\
  & = \bigcup_{U_\ell \in \mathcal{F}_1(X_\ell)} \{
  \langle f_2(u_1), \dots, f_2(u_\ell) \rangle \mid f_2 \in \mathcal{F}_2 \},
  \end{align*}
and therefore,
\begin{align*}
  N(\mathcal{F},X_\ell) 
  = |\mathcal{F}(X_\ell)| 
  & \leq 
    \sum_{U_\ell \in \mathcal{F}_1(X_\ell)} 
    \left|\{
          \langle f_2(u_1), \dots, f_2(u_\ell) \rangle 
           \mid f_2 \in \mathcal{F}_2 \}
    \right| 
    \\
    & \leq 
  \sum_{U_\ell \in \mathcal{F}_1(X_\ell)} 
  N(\mathcal{F}_2,U_\ell)
  \\
  & \leq 
  \sum_{U_\ell \in \mathcal{F}_1(X_\ell)} 
  G(\mathcal{F}_2,\ell)
  \\
  & =
  N(\mathcal{F}_1,X_\ell) \cdot
  G(\mathcal{F}_2,\ell)
  \\
  & \leq
  G(\mathcal{F}_1,\ell) \cdot
  G(\mathcal{F}_2,\ell).
\end{align*}
Since $X_\ell$ is arbitrary, this concludes the proof.
\end{proof}
%
% Proposition: Growth of parallel composition is a product
%
\begin{restatable}{proposition}{propgrowthparallelcomposition}
  \label{prop:growth-parallel-composition}
  Let $\mathcal{F}_1: X \to Y$ and let $\mathcal{F}_2: X \to Z$.
  Then,
  $
    G(\mathcal{F}_1 \times \mathcal{F}_2,\ell) 
    \leq
    G(\mathcal{F}_1,\ell) \cdot G(\mathcal{F}_2,\ell)
  $.
\end{restatable}
\begin{proof}
$$
G(\mathcal{F}_1 \times \mathcal{F}_2,\ell) 
= 
\sup_{X_\ell} N(\mathcal{F}_1 \times \mathcal{F}_2,X_\ell)
= 
\sup_{X_\ell} \big[N(\mathcal{F}_1,X_\ell) \cdot N(\mathcal{F}_2,X_\ell)\big]
\leq 
G(\mathcal{F}_1,\ell) \cdot G(\mathcal{F}_2,\ell).
$$
\end{proof}

The next Propositions~\ref{prop:stateless-growth}--\ref{prop:bar-growth}
describe properties of the growth with respect to string functions.
%
% Proposition: Growth of state-less functions
\begin{restatable}{proposition}{propstatelessgrowth}
  \label{prop:stateless-growth}
  Let $\mathcal{F}$ be a class of functions $\Sigma \to \Gamma$.
  Then, $G(\mathcal{F}^*,\ell) \leq G(\mathcal{F},\ell)$.
\end{restatable}
\begin{proof}
  Let $X_\ell = x_1, \dots, x_\ell$ be a sequence of strings over $\Sigma$, and
  let $\sigma_i$ be the last letter of $x_i$.
  Since $f^*(x_i) = f(\sigma_i)$ for every function $f \in \mathcal{F}$,
  we have that 
  $\mathcal{F}^*(X_\ell) = \mathcal{F}^*(\sigma_1, \dots, \sigma_\ell) =
  \mathcal{F}(\sigma_1, \dots, \sigma_\ell)$.
  Therefore,
  $$
    N(\mathcal{F}^*,X_\ell) 
    = |\mathcal{F}^*(X_\ell)|
    = |\mathcal{F}(\sigma_1, \dots, \sigma_\ell)|
    = N(\mathcal{F},\sigma_1, \dots, \sigma_\ell)
    \leq G(\mathcal{F},\ell).
  $$
  Since $X_\ell$ is arbitrary, this concludes the proof.
\end{proof}
%
% Proposition: Break bar of product
\begin{restatable}{proposition}{propparallelequigrowth}
  \label{prop:parallel-equigrowth}
  Let $\mathcal{F}_1: \Sigma^* \to Y$ and let $\mathcal{F}_2:\Sigma^* \to Z$.
  Then,
  $G(\overline{\mathcal{F}_1 \times \mathcal{F}_2},\ell) =
  G(\overline{\mathcal{F}_1} \times \overline{\mathcal{F}_2},\ell)$.
\end{restatable}
\begin{proof}
  Consider a function $f_1 \in \mathcal{F}_1$, 
  a function $f_2 \in \mathcal{F}_2$,
  and a sequence $X_\ell = x_1, \dots, x_\ell$ of strings over $\Sigma$.
  We have that
  $$
    (\overline{f_1 \times f_2})(X_\ell)
    = 
    (\overline{f_1 \times f_2})(x_1),
    \dots,
    (\overline{f_1 \times f_2})(x_\ell),
  $$
  and
  $$
    (\overline{f_1} \times \overline{f_2})(X_\ell)
    = 
    (\overline{f_1} \times \overline{f_2})(x_1),
    \dots,
    (\overline{f_1} \times \overline{f_2})(x_\ell).
   $$
  Now, consider $x_i = \sigma_1 \dots \sigma_s$.  We have
  $$
    (\overline{f_1 \times f_2})(x_i) 
    = 
    (f_1 \times f_2)(\sigma_1) \dots
    (f_1 \times f_2)(\sigma_1 \dots \sigma_s)
    = 
    \langle f_1(\sigma_1), f_2(\sigma_1) \rangle 
    \dots
    \langle f_1(\sigma_1 \dots \sigma_s), f_2(\sigma_1
    \dots \sigma_s) \rangle,
  $$
  and
  $$
    (\overline{f_1} \times \overline{f_2})(x_i) 
    = \langle \overline{f_1}(x_i), \overline{f_2}(x_i) \rangle
    = \langle f_1(\sigma_1) \dots f_1(\sigma_1 \dots \sigma_s), 
    f_2(\sigma_1) \dots f_2(\sigma_1 \dots \sigma_s) \rangle.
  $$
  Note that the output of $(\overline{f_1 \times f_2})(x_i)$ and 
  $(\overline{f_1} \times \overline{f_2})(x_i)$ contain the same occurrences of
  letters, arranged differently.  Hence, the same holds for
  $(\overline{f_1 \times f_2})(X_\ell)$
  and
  $(\overline{f_1} \times \overline{f_2})(X_\ell)$.
  Therefore,
  $$
    |\{(\overline{f_1 \times f_2})(X_\ell) \mid f_1 \in
    \mathcal{F}_1\setcomma f_2 \in \mathcal{F}_2 \}|
    = 
    |\{(\overline{f_1} \times \overline{f_2})(X_\ell) \mid f_1 \in
    \mathcal{F}_1\setcomma f_2 \in \mathcal{F}_2 \}|,
  $$
  and hence
  $$
    N(\overline{\mathcal{F}_1 \times \mathcal{F}_2},X_\ell) 
    = N(\overline{\mathcal{F}_1} \times \overline{\mathcal{F}_2},X_\ell).
  $$
  Since $X_\ell$ is arbitrary, this concludes the proof.
\end{proof}

%
% Proposition: Growth of Bared functions
\begin{proposition} \label{prop:bar-growth}
  Let $\mathcal{F}$ be a class of functions $\Sigma^* \to \Gamma$.
  Then,
$G(\overline{\mathcal{F}},\ell) \leq G(\mathcal{F},\ell \cdot M)$ where 
$M$ is the maximum length of a string.
\end{proposition}
\begin{proof}
  Consider a string $x = \sigma_1 \dots \sigma_s$ with 
  $\sigma_i \in \Sigma$.
  Consider the sequence of non-empty prefixes of $x$,
  $$x^\mathrm{p} = \sigma_1,\, \sigma_1 \sigma_2,\,
  \dots,\, \sigma_1 \sigma_2 \dots \sigma_s.$$ 
  For every $f \in \mathcal{F}$, the following two identities hold:
  \begin{align}
    \label{eq:prop:bar-growth-1}
    \overline{f}(x) & = f(\sigma_1) \dots f(\sigma_1 \dots
    \sigma_s),
    \\
    \label{eq:prop:bar-growth-2}
    f(x^\mathrm{p}) & = f(\sigma_1), \dots, f(\sigma_1 \dots
    \sigma_s).
  \end{align}
  Now, let $X_\ell = x_1, \dots, x_\ell$ be a sequence of $\ell$ strings over
  $\Sigma$, and let 
  $X^\mathrm{p}_\ell = x_1^\mathrm{p}, \dots, x_\ell^\mathrm{p}$ where
  $x_i^\mathrm{p}$ is the sequence of non-empty prefixes of $x_i$ defined as
  above.  
  Then, the identities~\eqref{eq:prop:bar-growth-1}
  and~\eqref{eq:prop:bar-growth-2} imply
  $N(\overline{\mathcal{F}},X_\ell) = N(\mathcal{F},X^\mathrm{p}_\ell)$.
  Thus,
  $N(\mathcal{F},X^\mathrm{p}_\ell) \leq G(\mathcal{F},\ell \cdot M)$,
  since $X^\mathrm{p}_\ell$ is a sequence of length at most $\ell \cdot M$,
  and hence
  $N(\overline{\mathcal{F}},X^\mathrm{p}_\ell) \leq G(\mathcal{F},\ell \cdot
  M)$.
  Since $X_\ell$ is arbitrary, we conclude that 
  $G(\overline{\mathcal{F}},\ell) \leq G(\mathcal{F},\ell \cdot
  M)$.
\end{proof}

The next Propositions~\ref{prop:bar}--\ref{prop:distributivity} provide
some algebraic identities for functions expressed in terms of composition, cross
product, and our constructors introduced in
Definition~\ref{def:function-constructors}.  
%
% Proposition: Break bar of composition
\begin{restatable}{proposition}{propbar}
  \label{prop:bar}
  Let $f_1: \Sigma^* \to \Gamma$, and 
  let $f_2: \Gamma^* \to Y$.
  Then,
  $\overline{\overline{f_1} \circ f_2} = \overline{f_1} \circ
  \overline{f_2}$.
\end{restatable}
\begin{proof}
  \begin{align*}
    & (\overline{\overline{f_1} \circ f_2})(\sigma_1 \dots \sigma_n)
    \\
    & = (\overline{f_1} \circ f_2)(\sigma_1) \dots
    (\overline{f_1} \circ f_2)(\sigma_1 \dots \sigma_n)
    \\
    & = f_2\big(\overline{f_1}(\sigma_1)\big) \dots
    f_2\big(\overline{f_1}(\sigma_1 \dots \sigma_n)\big)
    \\
    & = f_2\big(f_1(\sigma_1)\big) \dots
    f_2\big(f_1(\sigma_1) \dots f_1(\sigma_1 \dots \sigma_n)\big)
    \\
    & = \overline{f_2}\big(f_1(\sigma_1) \dots f_1(\sigma_1 \dots \sigma_n)\big)
    \\
    & = \overline{f_2}\big(\overline{f_1}(\sigma_1 \dots \sigma_n)\big)
    \\
    & = (\overline{f_1} \circ \overline{f_2})(\sigma_1 \dots \sigma_n).
  \end{align*}
\end{proof}
\begin{restatable}{proposition}{propstatelessbar}
  \label{prop:stateless-bar}
  Let $f_1: \Sigma^* \to \Gamma$, and 
  let $f_2: \Gamma \to Y$.
  Then,
  $\overline{f_1 \circ f_2} = \overline{f_1} \circ
  \overline{f_2^*}$.
\end{restatable}
\begin{proof}
  \begin{align*}
    & (\overline{f_1 \circ f_2})(\sigma_1 \dots \sigma_n)
    \\
    & = (f_1 \circ f_2)(\sigma_1) \dots
      (f_1 \circ f_2)(\sigma_1 \dots \sigma_n)
    \\
    & = f_2\big(f_1(\sigma_1)\big) \dots
    f_2\big(f_1(\sigma_1 \dots \sigma_n)\big)
    \\
    & = f^*_2\big(f_1(\sigma_1)\big) \dots
    f^*_2\big(f_1(\sigma_1 \dots \sigma_n)\big)
    \\
    & = f^*_2\big(f_1(\sigma_1)\big) \dots
    f^*_2\big(f_1(\sigma_1) \dots f_1(\sigma_1 \dots \sigma_n)\big)
    \\
    & = \overline{f^*_2}\big(f_1(\sigma_1) \dots f_1(\sigma_1 \dots
    \sigma_n)\big)
    \\
    & = \overline{f_2^*}\big(\overline{f_1}(\sigma_1 \dots \sigma_n)\big)
    \\
    & = (\overline{f_1} \circ \overline{f^*_2})(\sigma_1 \dots \sigma_n).
  \end{align*}
\end{proof}

\begin{proposition}[Left-distributivity of function \emph{composition} over
  \emph{cross product}]
  \label{prop:distributivity}
  Consider three 
  functions $f_1:X \to Y$, $f_2: Y \to Z$, and $f_3: Y \to W$.
  Then,
  $(f_1 \circ f_2) \times (f_1 \circ f_3) = f_1 \circ (f_2 \times f_3)$.
\end{proposition}
\begin{proof}
  We have
  \begin{align*}
    & [(f_1 \circ f_2) \times (f_1 \circ f_3)](x)
    \\
    & = \big\langle (f_1 \circ f_2)(x), (f_1 \circ f_3)(x) \big\rangle
    \\
    & = \big\langle f_2(f_1(x)), f_3(f_1(x)) \big\rangle
    \\
    & = (f_2 \times f_3)(f_1(x))
    \\
    & = (f_1 \circ (f_2 \times f_3))(x).
  \end{align*}
  This concludes the proof.
\end{proof}

\subsection{Proof of
Lemma~\ref{lemma:function-composition-generalised-automaton} (Functional
Description of an Automaton)}

\lemmafunctioncompositiongeneralisedautomaton*
\begin{proof}
  The function implemented by $A$ is the one implemented by the induced
  automaton 
  $A' = \langle \Sigma, Q, \delta_{J,\phi}, q_\mathrm{init}, \Gamma, \theta_J
  \rangle$, where
  $\delta_{J,\phi}(q,\sigma) =
  \delta(q,\phi(\pi_J(\sigma)))$, and the output function is 
  $\theta_J(q,\sigma) = \theta(q,\pi_J(\sigma))$.
  In turn, the function implemented by $A'$ is
  $A'(\sigma_1 \dots \sigma_m) =
  \theta_J(\delta_{J,\phi}(q_\mathrm{init},\sigma_1 \dots
  \sigma_{m-1}),\sigma_m)$.
  Thus,
  the function implemented by $A$ is
  \begin{align} \label{eq:lemmafunctioncompositiongeneralisedautomaton-1}
    A(\sigma_1 \dots \sigma_m) & =
  \theta_J(\delta_{J,\phi}(q_\mathrm{init},\sigma_1 \dots
  \sigma_{m-1}),\sigma_m).
  \end{align}
  Observe that
  \begin{align*}
    \phi(\pi_J(\sigma_1)) \dots \phi(\pi_J(\sigma_m))
    = \overline{\phi^*}(\pi_J(\sigma_1) \dots \pi_J(\sigma_m))
    = \overline{\phi^*}(\overline{\pi_J^*}(\sigma_1 \dots \sigma_m)).
  \end{align*}
  Thus,
  \begin{align*}
    \delta_{J,\phi}(q_\mathrm{init},\sigma_1 \dots \sigma_m) 
    = \delta(q_\mathrm{init}, \overline{\phi^*}(\overline{\pi_J^*}(\sigma_1 \dots
    \sigma_m))) 
    = D(\overline{\phi^*}(\overline{\pi_J^*}(\sigma_1 \dots \sigma_m))).
  \end{align*}
  Now, 
  going back to \eqref{eq:lemmafunctioncompositiongeneralisedautomaton-1},
  we have
  \begin{align*}
    A(\sigma_1 \dots \sigma_m)  
    & = \theta_J\big(\delta_{J,\phi}(q_\mathrm{init},\sigma_1 \dots
    \sigma_{m-1}),\sigma_m\big)
    \\
    & =
    \theta_J\big(D(\overline{\phi^*}(\overline{\pi_J^*}(\sigma_1 \dots
    \sigma_{m-1}))),\sigma_m\big)
    \\
    & =
    \theta_J\big(D^\popfunc(\overline{\phi^*}(\overline{\pi_J^*}(\sigma_1 \dots
    \sigma_m))),\sigma_m\big)
    \\
    & =
    \theta\big(D^\popfunc(\overline{\phi^*}(\overline{\pi_J^*}(\sigma_1 \dots
    \sigma_m))),\pi_J(\sigma_m)\big)
    \\
    & =
    \theta\big(D^\popfunc(\overline{\phi^*}(\overline{\pi_J^*}(\sigma_1 \dots
    \sigma_m))),\pi_J^*(\sigma_1 \dots \sigma_m)\big)
    \\
    & =
    \theta\big(D^\popfunc(\overline{\phi^*}(\overline{\pi_J^*}(\sigma_1 \dots
    \sigma_m))),I^*(\overline{\pi_J^*}(\sigma_1 \dots \sigma_m))\big)
    \\
    & =
    \theta\big(D^\popfunc((\overline{\pi_J^*} \circ \overline{\phi^*})(\sigma_1
    \dots \sigma_m)),(\overline{\pi_J^*} \circ I^*)(\sigma_1 \dots
    \sigma_m)\big)
    \\
    & =
    \theta\big((\overline{\pi_J^*} \circ \overline{\phi^*} \circ
    D^\popfunc)(\sigma_1 \dots \sigma_m),(\overline{\pi_J^*} \circ I^*)(\sigma_1
    \dots \sigma_m)\big)
    \\
    & =
    \theta\big(((\overline{\pi_J^*} \circ \overline{\phi^*} \circ
    D^\popfunc) \times (\overline{\pi_J^*} \circ I^*))(\sigma_1 \dots
    \sigma_m)\big)
    \\
    & =
    \big(\big((\overline{\pi_J^*} \circ \overline{\phi^*} \circ D^\popfunc)
    \times (\overline{\pi_J^*} \circ I^*) \big) \circ \theta\big)
    (\sigma_1 \dots \sigma_m)
    \\
    & =
    \big(\overline{\pi_J^*} \circ \big((\overline{\phi^*} \circ D^\popfunc)
    \times I^* \big) \circ \theta\big)(\sigma_1 \dots \sigma_m).
  \end{align*}
  The last equality is by Propostion~\ref{prop:distributivity}.
  This concludes the proof.
\end{proof}

\subsection{Proof of Theorem~\ref{theorem:vc-generalised-automata} (Growth and
Sample Complexity of Infinite Classes of Automata)}

%%%%
%
% Theorem Sample Complexity of Infinite Class of Automata
%
We prove Theorem~\ref{theorem:vc-generalised-automata} in two lemmas.
First, we prove the bound on the growth in 
Lemma~\ref{lemma:growth-generalised-automata}, then we prove the bound on the
dimension in Lemma~\ref{lemma:vc-dimension-automota}.
Then, the asymptotic bound on the sample complexity follows immediately by
Propostion~\ref{prop:sample-complexity-derivation}.

%
% Lemma: Growth of Automata as function composition
\begin{lemma}
  \label{lemma:growth-generalised-automata}
  Let $\mathcal{A}$ be a class of automata
  $\mathcal{A}(\Phi,\Delta,\Theta;a,m,\Pi,\Gamma)$, and
  let $M$ be the maximum length of a string.
  The growth of $\mathcal{A}$ is bounded as:
  $$
  G(\mathcal{A},\ell) 
  \leq 
  \left|\pi^a_m\right|
  \cdot 
  \left|\Delta \right|
  \cdot 
  G(\Phi,\ell \cdot M) 
  \cdot 
  G(\Theta,\ell).
  $$
\end{lemma}
\begin{proof}
  By Lemma~\ref{lemma:function-composition-generalised-automaton}
  the class of functions implemented by $\mathcal{A}$ can be expressed as
  \begin{equation*}
    \mathcal{A} 
    = 
    \overline{(\pi_m^a)^*} 
     \circ 
     \left(
         (\overline{\Phi^*} \circ \Delta^{\popfunc}) \times I^*
     \right)
     \circ 
     \Theta.  
  \end{equation*}
Then,
  \begin{align*}
    G(\mathcal{A},\ell) 
    & = 
    G\left(
        \overline{(\pi_m^a)^*} 
         \circ 
         \left(
             (\overline{\Phi^*} \circ \Delta^{\popfunc}) \times I^*
         \right)
         \circ 
         \Theta, \ell
    \right)
    &\mbox{ by Lemma~\ref{lemma:function-composition-generalised-automaton},}
  \\
    &\leq 
    G\left(
        \overline{(\pi_m^a)^*} 
         \circ 
         \left(
             (\overline{\Phi^*} \circ \Delta^{\popfunc}) \times I^*
         \right), \ell 
    \right)
    \cdot 
    G(\Theta,\ell)
    &\mbox{by Proposition~\ref{prop:growth-composition},}
  \\
    &\leq 
    G\left(\overline{(\pi_m^a)^*}, \ell\right)
    \cdot 
    G\left(
        (\overline{\Phi^*} \circ \Delta^{\popfunc}) \times I^*, \ell 
    \right)
    \cdot
    G(\Theta,\ell)
    &\mbox{by Proposition~\ref{prop:growth-composition},}
  \\
    & \leq
    \left|\pi_m^a\right|
    \cdot 
    G\left(
        (\overline{\Phi^*} \circ \Delta^{\popfunc}) \times I^*, \ell 
    \right)
    \cdot
    G(\Theta,\ell)
    &\mbox{since $\pi_m^a$ is finite,}
  \\
    & \leq
    \left|\pi_m^a\right|
    \cdot 
    G\left(\overline{\Phi^*} \circ \Delta^{\popfunc}, \ell \right)
    \cdot 
    G\left(I^*, \ell \right)
    \cdot 
    G(\Theta,\ell)
    &\mbox{by Proposition~\ref{prop:growth-parallel-composition},}
  \\
    & =
    \left|\pi_m^a\right|
    \cdot 
    G\left(\overline{\Phi^*} \circ \Delta^{\popfunc}, \ell \right)
    \cdot 
    G(\Theta,\ell)
    &\mbox{since $I^*$ is a single function,}
  \\
    & \leq
    \left|\pi_m^a\right|
    \cdot 
    G\left(\overline{\Phi^*}, \ell\right)
    \cdot 
    G\left(\Delta^{\popfunc}, \ell \right)
    \cdot 
    G\left(\Theta,\ell\right)
    &\mbox{by Proposition~\ref{prop:growth-composition},}
  \\
    & \leq
    \left|\pi_m^a\right|
    \cdot 
    \left|\Delta\right|
    \cdot 
    G\left(\overline{\Phi^*}, \ell\right)
    \cdot 
    G\left(\Theta,\ell\right)
    &\mbox{since $\Delta$ is finite,}
  \\
    &\leq 
    \left|\pi_m^a\right|
    \cdot 
    \left|\Delta\right|
    \cdot 
    G(\Phi^*, \ell \cdot M) 
    \cdot 
    G(\Theta,\ell)
    &\mbox{ by Proposition~\ref{prop:bar-growth}, }
  \\
    &\leq 
    \left|\pi_m^a\right|
    \cdot 
    \left|\Delta\right|
    \cdot 
    G(\Phi, \ell \cdot M) 
    \cdot
    G(\Theta,\ell)
    &\mbox{ by Proposition~\ref{prop:stateless-growth}. }
  \end{align*}
  This proves the lemma.
\end{proof}

\begin{lemma}
  \label{lemma:vc-dimension-automota}
  Let $\mathcal{A}$ be a class of automata
  $\mathcal{A}(\Phi,\Delta,\Theta;a,m,\Pi,\Gamma)$, 
  let $M$ be the maximum length of a string, and
  let $w = \log\left|\pi_m^a\right| + \log\left|\Delta\right| +
  \dimension(\Phi) + \dimension(\Theta) \geq 2$.
  The dimension of $\mathcal{A}$ is bounded as:
  \begin{align*}
    \dimension(\mathcal{A}) 
    & 
    \leq 2 \cdot w \cdot \log(w \cdot e \cdot M \cdot |\Pi| \cdot |\Gamma|).
  \end{align*}
\end{lemma}
\begin{proof}
  By Lemma~\ref{lemma:growth-generalised-automata} the upper bound on the 
  growth $G(\mathcal{A},\ell)$ is 
  $$
    G(\mathcal{A},\ell) 
    \leq 
    \left|\pi^a_m\right|
    \cdot 
    \left|\Delta\right|
    \cdot 
    G(\Phi,\ell \cdot M) 
    \cdot 
    G(\Theta,\ell).
  $$
  Applying Proposition~\ref{prop:bound-growth-vc-haussler} 
  to
  $G(\Phi,\ell \cdot M)$ and $G(\Theta,\ell)$ we obtain,
  $$
    G(\mathcal{A},\ell) 
    \leq 
    \left|\pi_m^a\right| 
    \cdot
    \left|\Delta\right| 
    \cdot
    \left(e \cdot \ell \cdot M \cdot |\Pi|\right)^h
    \cdot  
    \left(e \cdot \ell \cdot |\Gamma|\right)^g.
  $$
We loosen the bound on $G(\mathcal{A},\ell)$ by collecting the factors with 
an exponent,
  $$
    G(\mathcal{A},\ell) 
    \leq 
    \left|\pi_m^a\right| 
    \cdot
    \left|\Delta\right| 
    \cdot 
    \left(
        \ell \cdot e \cdot M \cdot |\Pi| \cdot |\Gamma|
    \right)^{h + g}.
  $$
Next we show the claimed upper bound on the dimension of $\mathcal{A}$.
By the definition of dimension,
in order to show that a number $\ell$ is an upper bound on the dimension of
$\mathcal{A}$, 
it suffices to show that $G(\mathcal{A}_\mathrm{bin},\ell) < 2^\ell$---note that 
$G(\mathcal{A}_\mathrm{bin},\ell) = G(\mathcal{A},\ell)$ when the output
alphabet has cardinality two.
Then, since $G(\mathcal{A}_\mathrm{bin},\ell) \leq G(\mathcal{A},\ell)$ by
Proposition~\ref{prop:inequality-two-growths}, a number $\ell$
satisfies $G(\mathcal{A}_\mathrm{bin},\ell) < 2^\ell$ if it satisfies 
$G(\mathcal{A},\ell) < 2^\ell$, and in turn if it satisfies
  $$
    \left|\pi_m^a\right| 
    \cdot
    \left|\Delta\right| 
    \cdot 
    \left(
        \ell \cdot e \cdot M \cdot |\Pi| \cdot |\Gamma|
    \right)^{h + g}
    < 2^{\ell}.
  $$
We show that the former inequality is satisfied for
$
\ell 
= 
2 \cdot w \cdot \log(w \cdot e \cdot M \cdot |\Pi| \cdot |\Gamma|)
= 
2 \cdot w \cdot \log(w \cdot e \cdot C)
$,
where we introduce $C = M \cdot |\Pi| \cdot |\Gamma|$ for brevity.
We keep proceeding by sufficient conditions. Specifically, at each step, 
either
we replace the r.h.s.\ with a smaller expression, 
we simplify a common term between the two sides, or 
we just rewrite an expression in an equivalent form.
\begin{align*}
    \left|\pi_m^a\right| 
    \cdot
    \left|\Delta\right| 
    \cdot 
    \left(
        \ell \cdot e \cdot M \cdot |\Pi| \cdot |\Gamma|
    \right)^{h + g}
    &< 2^{\ell}
\\
    \left|\pi_m^a\right| 
    \cdot
    \left|\Delta\right| 
    \cdot 
    \left(
        \ell \cdot e \cdot C
    \right)^{h + g}
    &< 2^{\ell}
\\
    \left|\pi_m^a\right| 
    \cdot
    \left|\Delta\right| 
    \cdot 
    \big(
        2 \cdot w \cdot \log(w \cdot e \cdot C) \cdot e \cdot C
    \big)^{h + g}
    &< 
    2^{2 \cdot w \cdot \log(w \cdot e \cdot C)},
\\
    \left|\pi_m^a\right| 
    \cdot
    \left|\Delta\right| 
    \cdot 
    \big(
        2 \cdot w \cdot \log(w \cdot e \cdot C) \cdot e \cdot C
    \big)^{h + g}
    &< 
    (w \cdot e \cdot C)^{2 \cdot (h+g)} 
    \cdot 
    (w \cdot e \cdot C)^{
        2 \cdot 
        (\ln(\left|\Delta\right| \cdot \left|\pi_m^a\right|)/\ln 2))},
\\
    \left|\pi_m^a\right| 
    \cdot
    \left|\Delta\right| 
    \cdot 
    \big(
        2 \cdot w \cdot \log(w \cdot e \cdot C) \cdot e \cdot C
    \big)^{h + g}
    &< 
    (w \cdot e \cdot C)^{2 \cdot (h+g)} 
    \cdot 
    (w \cdot e \cdot C)^{
        2 \cdot (\ln(\left|\Delta\right| \cdot \left|\pi_m^a\right|))},
\\
    \left|\pi_m^a\right| 
    \cdot
    \left|\Delta\right| 
    \cdot 
    \big(
        2 \cdot w \cdot \log(w \cdot e \cdot C) \cdot e \cdot C
    \big)^{h + g}
    &< 
    (w \cdot e \cdot C)^{2 \cdot (h+g)} 
    \cdot 
    \left|\pi_m^a\right|
    \cdot
    \left|\Delta\right|,
\\
    \big(
        2 \cdot w \cdot \log(w \cdot e \cdot C) \cdot e \cdot C
    \big)^{h + g}
    &< 
    (w \cdot e \cdot C)^{2 \cdot (h+g)} 
\\
    2 \cdot \log(w \cdot e \cdot C)^{h + g}
    &< 
    (w \cdot e \cdot C)^{h+g} 
\\
    1 
    &< 
    \bigg(
        \frac{(w \cdot e \cdot C)}{2 \cdot \log(w \cdot e \cdot C)}
    \bigg)^{h+g}
\end{align*}
The term on the r.h.s.\ is strictly greater than $1$
whenever $w \geq 2$, considering that $C \geq 1$.
This proves the lemma.
\end{proof}

\subsection{Proof of Lemma~\ref{lemma:function-compoisition-of-cascade}
(Functional Description of an Automata Cascade)}

\lemmafunctioncompoisitionofcascade*
\begin{proof}
  Let 
  $C_d = A_1 \cascade \cdots \cascade A_d = (A_1, \dots, A_d)$.
  We show by induction on $d \geq 1$ that
  \begin{equation*}
    \overline{(I^* \times A_1)} \circ \dots \circ
    \overline{(I^* \times A_{d-1})} \circ A_d.
  \end{equation*}
  considering $C_1 = A_1$ as the base case, which holds trivially.
  In the inductive case, $d \geq 2$, and we assume that the function
  implemented by the cascade 
  $C_{d-1} = A_1 \cascade \cdots \cascade A_{d-1}$ can be expressed as
  \begin{equation*}
    \overline{(I^* \times A_1)} \circ \dots \circ
    \overline{(I^* \times A_{d-2})} \circ A_{d-1}.
  \end{equation*}
  From the definition of cascade, at each step, the input to component $A_d$
  consists of the input to the component $A_{d-1}$ together with the output of
  $A_{d-1}$.
  Over an execution, the sequence of such inputs to $A_d$ is the function
  \begin{equation*}
    \overline{(I^* \times A_1)} \circ \dots \circ
    \overline{(I^* \times A_{d-2})} \circ (\overline{I^* \times A_{d-1}}).
  \end{equation*}
  Then, the overall output of the casacade is returned by $A_d$ applied to such
  an input string, and hence the function of $C_d$ is
  \begin{equation*}
    \overline{(I^* \times A_1)} \circ \dots \circ
    \overline{(I^* \times A_{d-2})} \circ (\overline{I^* \times A_{d-1}}) \circ
    A_d.
  \end{equation*}
  This concludes the proof.
\end{proof}

\subsection{Proof of Theorem~\ref{theorem:v-dimension-cascades} (Growth and
Sample Complexity of Infinite Classes of Automata Cascades)}

%%%%
% Theorem Sample Complexity of Infinite Classes of Cascades
%
We prove Theorem~\ref{theorem:v-dimension-cascades} in two lemmas.
First, we prove the bound on the growth of cascades in 
Lemma~\ref{lemma:growth-cascades}, then we prove the bound on the
dimension of cascades in Lemma~\ref{lemma:vc-dimension-cascades}.
Then, the asymptotic bound on the sample complexity follows immediately by
Propostion~\ref{prop:sample-complexity-derivation}.

\begin{lemma}
  \label{lemma:growth-cascades}
  Let $\mathcal{C}$ be a class 
  $\mathcal{C} = \mathcal{A}_1 \cascade \dots \cascade \mathcal{A}_d$ 
  of cascades of automata
  $
  \mathcal{A}_i = 
  \mathcal{A}(\Phi_i,\Delta_i,\Theta_i;a_i,m_i,\Pi_i,\Gamma_i)
  $, 
  and let $M$ be the maximum length of a string.
  The growth of\/ $\mathcal{C}$ is bounded as:
  $$
    G(\mathcal{C}, \ell) 
    \leq 
    \prod_{i=1}^d 
    \left|\pi_{m_i}^{a_i}\right|
    \cdot 
    \left|\Delta_i \right| 
    \cdot
    G(\Phi_i,\ell \cdot M) 
    \cdot
    G(\Theta_i,\ell \cdot M).
  $$
\end{lemma}
\begin{proof}
  By Lemma~\ref{lemma:function-compoisition-of-cascade} the class of 
  functions implemented by $\mathcal{C}$ can be expressed as,
  $$
    \mathcal{C} 
    = 
    \overline{(I^* \times \mathcal{A}_1)} \circ \dots \circ
    \overline{(I^* \times \mathcal{A}_{d-1})} \circ A_d.
  $$
  Then,
  \begin{align*}
      G(\mathcal{C},\ell) 
      &= 
      G\left(
        \overline{(I^* \times \mathcal{A}_1)} \circ \dots \circ
        \overline{(I^* \times \mathcal{A}_{d-1})} \circ A_d, \ell
      \right)
      &\mbox{ by Lemma~\ref{lemma:function-compoisition-of-cascade},}
  \\
      &\leq
      G(\mathcal{A}_d, \ell) 
      \cdot 
      \prod_{i=1}^{d-1} 
      G(\overline{I^* \times \mathcal{A}_i},\ell)
      &\mbox{by Proposition~\ref{prop:growth-composition},}
  \\
      &\leq
      G(\mathcal{A}_d, \ell) \cdot
      \prod_{i=1}^{d-1}
      G(\overline{I^*},\ell) \cdot 
      G(\overline{\mathcal{A}_i},\ell)
      &\mbox{by Proposition~\ref{prop:growth-parallel-composition}
             and~\ref{prop:parallel-equigrowth},}
  \\
      &=
      G(\mathcal{A}_d, \ell) \cdot
      \prod_{i=1}^{d-1}
      G(\overline{\mathcal{A}_i},\ell)
      &\mbox{ since $I^*$ is singleton,}
  \end{align*}
  Next we derive an upper bound on 
  $G(\overline{\mathcal{A}},\ell)$ 
  for an arbitrary class
  $\mathcal{A} = \mathcal{A}(\Phi,\Delta,\Theta;a,m)$. 
  \begin{align*}
    G(\overline{\mathcal{A}}, \ell)
    &= 
    G\left(
      \overline{
        \overline{(\pi_m^a)^*} 
         \circ 
         \left(
             (\overline{\Phi^*} \circ \Delta^{\popfunc}) \times I^*
         \right)
         \circ 
         \Theta}, \ell
    \right)
    &\mbox{ by Lemma~\ref{lemma:function-composition-generalised-automaton},}
\\
    &\leq
    G\left(
        \overline{
            \overline{(\pi_m^a)^*} 
             \circ 
             \left(
                 (\overline{\Phi^*} \circ \Delta^{\popfunc}) \times I^*
             \right)},\ell 
    \right)
    \cdot
    G\left(\overline{\Theta^*}, \ell \right)
    &\mbox{ by Proposition~\ref{prop:stateless-bar} and
            \ref{prop:growth-composition},}
\\
    &\leq
    G\left(
        \overline{(\pi_m^a)^*} 
         \circ 
         \big(
            \overline{(\overline{\Phi^*} \circ \Delta^{\popfunc}) \times I^*}
        \big),\ell 
    \right)
    \cdot
    G\left(\overline{\Theta^*}, \ell \right)
    &\mbox{ by Proposition~\ref{prop:bar},}
\\
    &\leq
    G\left(\overline{(\pi_m^a)^*}, \ell\right)
    \cdot 
    G\left(
        \overline{(\overline{\Phi^*} \circ \Delta^{\popfunc}) \times I^*},\ell 
    \right)
    \cdot
    G\left(\overline{\Theta^*}, \ell \right)
    &\mbox{ by Proposition~\ref{prop:growth-composition},}
\\
    &\leq
    \left|\pi_m^a\right|
    \cdot 
    G\left(
        \overline{(\overline{\Phi^*} \circ \Delta^{\popfunc}) \times I^*},\ell 
    \right)
    \cdot
    G\left(\overline{\Theta^*}, \ell \right)
    &\mbox{since $\pi_m^a$ is finite,}
\\
    &\leq
    \left|\pi_m^a\right|
    \cdot 
    G\left(
        \overline{(\overline{\Phi^*} \circ \Delta^{\popfunc})}
        \times 
        \overline{I^*},\ell 
    \right)
    \cdot
    G\left(\overline{\Theta^*}, \ell \right)
     &\mbox{ by Proposition~\ref{prop:parallel-equigrowth},}
\\
    &\leq
    \left|\pi_m^a\right|
    \cdot 
    G\left(
        \overline{(\overline{\Phi^*} \circ \Delta^{\popfunc})}, \ell
    \right)
    \cdot 
    G\left(\overline{I^*},\ell \right)
    \cdot
    G\left(\overline{\Theta^*}, \ell \right)
      &\mbox{ by Proposition~\ref{prop:growth-parallel-composition},}
\\
    &\leq
    \left|\pi_m^a\right|
    \cdot 
    G\left(
        \overline{(\overline{\Phi^*} \circ \Delta^{\popfunc})}, \ell
    \right)
    \cdot
    G\left(\overline{\Theta^*}, \ell \right)
    &\mbox{ since $I^*$ is a single function,}
\\
    &\leq
      \left|\pi_m^a\right|
      \cdot 
      G\left(\overline{\Phi^*} \circ \overline{\Delta^{\popfunc}},\ell\right)
      \cdot 
      G\left(\overline{\Theta^*},\ell\right)
      &\mbox{ by Proposition~\ref{prop:bar},}
\\
    &\leq
      \left|\pi_m^a\right|
      \cdot 
      G\left(\overline{\Phi^*}, \ell \right)
      \cdot 
      G\left(\overline{\Delta^{\popfunc}},\ell\right)
      \cdot 
      G\left(\overline{\Theta^*},\ell\right)
      &\mbox{ by Proposition~\ref{prop:growth-composition},}
  \\
    &\leq
      \left|\pi_m^a\right|
      \cdot 
      \left|\Delta\right|
      \cdot 
      G\left(\overline{\Phi^*}, \ell\right)
      \cdot 
      G\left(\overline{\Theta^*},\ell\right)
      &\mbox{ since $\Delta$ is a finite class,}
  \\
    &\leq
      \left|\pi_m^a\right|
      \cdot 
      \left|\Delta\right|
      \cdot 
      G\left(\Phi^*,\ell \cdot M\right) 
      \cdot
      G\left(\Theta^*,\ell \cdot M\right)
      &\mbox{ by Proposition~\ref{prop:bar-growth},}
  \\
    &\leq
      \left|\pi_m^a\right|
      \cdot 
      \left|\Delta\right|
      \cdot 
      G\left(\Phi,\ell \cdot M\right) 
      \cdot
      G\left(\Theta,\ell \cdot M\right)
      &\mbox{ by Proposition~\ref{prop:stateless-growth}.}
  \end{align*}
Using Lemma~\ref{lemma:growth-generalised-automata} we derive an upper bound
on the $G(\mathcal{A}_d,\ell)$,
$$
  G(\mathcal{A}_d,\ell) 
   \leq 
    \left|\pi_{m_d}^{a_d}\right| 
    \cdot 
    \left|\Delta_d\right| 
    \cdot 
    G(\Phi_d,\ell \cdot M) 
    \cdot 
    G(\Theta_d,\ell)
   \leq 
    \left|\pi_{m_d}^{a_d}\right| 
    \cdot 
    \left|\Delta_d\right| 
    \cdot 
    G(\Phi_d,\ell \cdot M) 
    \cdot 
    G(\Theta_d,\ell \cdot M).
$$
Then,
$$
    G(\mathcal{C},\ell) 
    \leq 
    \prod_{i=1}^{d}
    \left|\pi_{m_i}^{a_i}\right| 
    \cdot 
    \left|\Delta_i\right| 
    \cdot 
    G(\Phi_i,\ell \cdot M) 
    \cdot 
    G(\Theta_i,\ell \cdot M).
$$
The lemma is proved.
\end{proof}

\begin{lemma}
    \label{lemma:vc-dimension-cascades}
    Let $\mathcal{C}$ be a class 
    $\mathcal{C} = \mathcal{A}_1 \cascade \dots \cascade \mathcal{A}_d$ 
    of cascades where
    $
    \mathcal{A}_i = 
    \mathcal{A}(\Phi_i,\Delta_i,\Theta_i;a_i,m_i,\Pi_i,\Gamma_i)
    $.
    Let $M$ be the maximum length of a string.
    Let $h = \max_i \dimension(\Phi_i)$ 
    be maximum dimension of a class of input functions,
    let $g = \max_i \dimension(\Theta_i)$ be the maximum dimension of a class of
    output functions,
    let $\Delta = \argmax_{\Delta_i} \left|\Delta_i\right|$ be the class of
    semiautomata with maximum cardinality, and
    let $\pi_m^a = \argmax_{\pi_{m_i}^{a_i}} \left|\pi_{m_i}^{a_i}\right|$.
    Let $w = \log\left|\pi_m^a\right| + \log\left|\Delta\right| + h +
    g \geq 2$.
    The dimension of\/ $\mathcal{C}$ is bounded as:
    $$
      \dimension(\mathcal{C}) 
      \leq 
      2 \cdot d \cdot w 
      \cdot \log(d \cdot w \cdot e \cdot M \cdot |\Pi| \cdot |\Gamma|),
    $$
  \end{lemma}
\begin{proof}
  By Lemma~\ref{lemma:growth-cascades} the upper bound on 
  $G(\mathcal{C}, \ell)$ is 
  \begin{align*}
    G(\mathcal{C}, \ell) 
    \leq 
    \prod_{i=1}^{d}
    \left|\pi_{m_i}^{a_i}\right| 
    \cdot 
    \left|\Delta_i\right| 
    \cdot 
    G(\Phi_i,\ell \cdot M) 
    \cdot 
    G(\Theta_i,\ell \cdot M).
  \end{align*}
  Applying Proposition~\ref{prop:bound-growth-vc-haussler} 
  to terms
  $G(\Phi_i,\ell \cdot M),  G(\Theta_i,\ell \cdot M)$ we obtain
  $$
    G(\mathcal{C}, \ell) 
    \leq 
    \prod_{i=1}^d
    \left|\pi_{m_i}^{a_i}\right|
    \cdot
    |\Delta_i| 
    \cdot  
    \left(e \cdot \ell \cdot M \cdot |\Pi_i|\right)^{h_i} 
    \cdot
    \left(e \cdot \ell \cdot |\Gamma_i|\right)^{g_i}.
  $$
We loosen the bound by collecting the factors with the exponent 
and by taking the maximum of each indexed quantity,
$$
    G(\mathcal{C}, \ell) 
    \leq 
    \left|\pi_m^a\right|^d
    \cdot 
    \left|\Delta\right|^d 
    \cdot  
    \left(
     \ell \cdot e \cdot M \cdot |\Pi| \cdot |\Gamma|
    \right)^{d \cdot (h+g)} 
$$
Next we show the claimed upper bound on the dimension of $\mathcal{C}$.
By the definition of dimension,
in order to show that a number $\ell$ is an upper bound on the dimension, 
it suffices to show that $G(\mathcal{C}_\mathrm{bin},\ell) < 2^\ell$. 
Then, since $G(\mathcal{C}_\mathrm{bin},\ell) \leq G(\mathcal{C},\ell)$ by
Proposition~\ref{prop:inequality-two-growths}, a number $\ell$ satisfies 
$G(\mathcal{C}_\mathrm{bin},\ell) < 2^\ell$ if it satisfies 
$G(\mathcal{C},\ell) < 2^\ell$, and in turn if it satisfies
$$
    \left(\left|\pi_m^a\right| \cdot \left|\Delta\right|\right)^d 
    \cdot  
    \left(
     \ell \cdot e \cdot M \cdot |\Pi| \cdot |\Gamma|
    \right)^{d \cdot (h+g)} 
    < 
    2^{\ell}.
$$
We show that the former inequality is satisfied for 
$
\ell 
= 2 \cdot d \cdot w \cdot 
  \log( d \cdot w \cdot e \cdot M \cdot |\Pi| \cdot |\Gamma|)
= 2 \cdot d \cdot w \cdot 
  \log( d \cdot w \cdot e \cdot C),
$
where we introduce $C = M \cdot |\Pi| \cdot |\Gamma|$ for brevity.
We keep proceeding by sufficient conditions. 
Specifically, at each step, either
we replace the r.h.s.\ with a smaller expression, 
we simplify a common term between the two sides,
or we just rewrite an expression in an equivalent form.
\begin{align*}
    \left(\left|\pi_m^a\right| \cdot \left|\Delta\right|\right)^d 
    \cdot  
    \left(
     \ell \cdot e \cdot M \cdot |\Pi| \cdot |\Gamma|
    \right)^{d \cdot (h+g)} 
    &< 
    2^{\ell}
\\
    \left(\left|\pi_m^a\right| \cdot \left|\Delta\right|\right)^d 
    \cdot  
    \left( \ell \cdot e \cdot C \right)^{d \cdot (h+g)} 
    &< 
    2^{\ell}
\\
    \left(\left|\pi_m^a\right| \cdot \left|\Delta\right|\right)^d 
    \cdot  
    \big(
         2 \cdot d \cdot w \cdot \log(d \cdot w \cdot e \cdot  C)
         \cdot e \cdot C
    \big)^{d \cdot (h+g)} 
    &< 
    2^{
         2 \cdot d \cdot w \cdot \log(d \cdot w \cdot e \cdot C)
      },
\\
    \left(\left|\pi_m^a\right| \cdot \left|\Delta\right|\right)^d 
    \cdot  
    \big(
         2 \cdot d \cdot w \cdot \log(d \cdot w \cdot e \cdot C)
         \cdot e \cdot C
    \big)^{d \cdot (h+g)} 
    &< 
    (d \cdot w \cdot e \cdot C)^{2 \cdot d \cdot w},
\\
    \left(\left|\pi_m^a\right| \cdot \left|\Delta\right|\right)^d 
    \cdot  
    \big(
         2 \cdot d \cdot w \cdot \log(d \cdot w \cdot e \cdot C)
         \cdot e \cdot C
    \big)^{d \cdot (h+g)} 
    &< 
    (d \cdot w \cdot e \cdot C)^{2 \cdot d \cdot (h + g)}
    \cdot 
    (d \cdot w \cdot e \cdot C)^{
            2 \cdot d \cdot (\log(|\Delta| \cdot |\pi_m^a|))},
\\
    \left(\left|\pi_m^a\right| \cdot \left|\Delta\right|\right)^d 
    \cdot  
    \big(
         2 \cdot d \cdot w \cdot \log(d \cdot w \cdot e \cdot C)
         \cdot e \cdot C
    \big)^{d \cdot (h+g)} 
    &< 
    (d \cdot w \cdot e \cdot C)^{2 \cdot d \cdot (h + g)}
    \cdot 
    (d \cdot w \cdot e \cdot C)^{
        2 \cdot d \cdot (\ln(|\Delta|\cdot |\pi_m^a|)/\ln2)},
\\
    \left(\left|\pi_m^a\right| \cdot \left|\Delta\right|\right)^d 
    \cdot  
    \big(
         2 \cdot d \cdot w \cdot \log(d \cdot w \cdot e \cdot C)
         \cdot e \cdot C
    \big)^{d \cdot (h+g)} 
    &< 
    (d \cdot w \cdot e \cdot C)^{2 \cdot d \cdot (h + g)}
    \cdot 
    (d \cdot w \cdot e \cdot C)^{
        2 \cdot d \cdot \ln(|\Delta| \cdot |\pi_m^a|)},
\\
    \left(\left|\pi_m^a\right| \cdot \left|\Delta\right|\right)^d 
    \cdot  
    \big(
         2 \cdot d \cdot w \cdot \log(d \cdot w \cdot e \cdot C)
         \cdot e \cdot C
    \big)^{d \cdot (h+g)} 
    &< 
    (d \cdot w \cdot e \cdot C)^{2 \cdot d \cdot (h + g)}
    \cdot 
    \left(\left|\pi_m^a\right| \cdot \left|\Delta\right|\right)^{2 \cdot d},
\\
    \big(
         2 \cdot d \cdot w \cdot \log(d \cdot w \cdot e \cdot C)
         \cdot e \cdot C
    \big)^{d \cdot (h+g)} 
    &< 
    (d \cdot w \cdot e \cdot C)^{2 \cdot d \cdot (h + g)},
\\
     2 \cdot \log(d \cdot w \cdot e \cdot C)^{d \cdot (h+g)} 
    &< 
    (d \cdot w \cdot e \cdot C)^{d \cdot (h + g)},
\\
    1 
    &< 
    \left(
        \frac{(d \cdot w \cdot e \cdot C)}
             {2 \cdot \log(d \cdot w \cdot e \cdot C)}
    \right)^{d \cdot (h+g)}.
\end{align*}
The term on the r.h.s.\ is strictly greater than $1$ 
whenever $d \cdot w \geq 2$, considering that $C \geq 1$. The lemma is proved.
\end{proof}

\section{Formal Description of The Running Example}

The sequence task of the running example (Example~\ref{ex:running-example-1}) is
formalised by the following Temporal Datalog program---for Temporal Datalog see 
(Ronca et al.\ 2022).
Let 
$\mathit{wood}$,
$\mathit{iron}$,
$\mathit{fire}$,
$\mathit{steel}$,
and
$\mathit{factory}$ be unary
unary predicates.
They are used to model the input.
For example, $\mathit{wood}(17)$ is in the input trace if the agent collects
wood at time $17$, and $\mathit{factory}(5)$ is in the input trace if the agent
attempts to use the factory at time $5$.
Additionally, let 
$\mathit{getWood}$,
$\mathit{getIron}$,
$\mathit{getFire}$,
$\mathit{getSteel}$,
and
$\mathit{useFactory}$
be unary predicates.
They are defined by the following rules.
\begin{align}
  \mathit{wood}(t)  & \to \mathit{getWood}(t)  
  \\
  \mathit{getWood}(t)  & \to \mathit{getWood}(t+1)  
  \\
  \mathit{iron}(t)  & \to \mathit{getIron}(t)  
  \\
  \mathit{getIron}(t)  & \to \mathit{getIron}(t+1)  
  \\
  \mathit{fire}(t)  & \to \mathit{getFire}(t)  
  \\
  \mathit{getFire}(t)  & \to \mathit{getFire}(t+1)  
  \\
  \mathit{steel}(t)  & \to \mathit{getSteel}(t)  
  \\
  \mathit{getSteel}(t)  & \to \mathit{getSteel}(t+1)  
  \\
  \mathit{getSteel}(t-1) \land \mathit{factory}(t)  & \to \mathit{useFactory}(t)  
  \\
  \mathit{getWood}(t-1) \land \mathit{getIron}(t-1) \land \mathit{getFire}(t-1)
  \land \mathit{factory}(t)  & \to \mathit{useFactory}(t)  
  \\
  \mathit{useFactory}(t)  & \to \mathit{useFactory}(t+1)  
\end{align}
Then, at any given time point $\tau$,
the task has been completed if $\mathit{useFactory}(\tau)$ is entailed by the
input facts and the rules above.

\section*{References for The Appendix}

\smallskip
\par
\noindent
Kakade, S.; and Tewari, A. 2008.
\newblock Lecture Notes of CMSC 35900 (Spring 2008) Learning
  Theory, Lecture 12.
  \newblock Available at: 
  \url{https://home.ttic.edu/%7Etewari/lectures/lecture12.pdf}
\newblock [Last accessed: 30/11/2022].

\smallskip
\par
\noindent
Ronca, A.; Kaminski, M.; Cuenca Grau, B.; and Horrocks, I. 2022.
\newblock The delay and window size problems in rule-based stream reasoning.
\newblock \emph{Artif. Intell.}, 306: 103668.

% \smallskip
% \par
% \noindent
% Shawe-Taylor, J.; and Anthony, M. 1991.
% \newblock Sample sizes for multiple-output threshold networks.
% \newblock \emph{Network: Computation in Neural Systems}, 2(1): 107--117.

\fi

\end{document}